\documentclass[journal]{IEEEtran} 

\usepackage{graphics,graphicx} 
\usepackage{epsfig} 
\usepackage{amsmath} 
\usepackage{amssymb}  
\usepackage{amsfonts}
\usepackage{amsthm}
\usepackage{mathtools}  
\usepackage{optidef,booktabs}
\usepackage{cite}
\usepackage{url}
\usepackage{lipsum} 
\usepackage{algorithm} 
\usepackage{algpseudocode}
\usepackage[utf8]{inputenc} 
\usepackage[font=small,skip=0pt]{subcaption}
\usepackage[font=small,skip=0pt]{caption}
\usepackage{tabularray}
\usepackage{tikz}
\usepackage{circuitikz}
\usepackage{xcolor}
\usepackage{nicefrac}
\usepackage{hyperref} 
\hypersetup{ 		
	colorlinks=true,        
	linkcolor=blue,
	urlcolor=black,
	citecolor=blue    		
}

\UseTblrLibrary{amsmath}

\theoremstyle{remark}
\newtheorem{remark}{Remark}
\newtheorem{example}{Example}
\newtheorem{lemma}{Lemma}
\theoremstyle{definition}
\newtheorem{theorem}{Theorem}

\newtheorem{definition}{Definition}
\newtheorem{fact}{Property}

\newtheorem{assumption}{\textbf{Assumption}}

 \definecolor{color1bg}{HTML}{77AC30}
\definecolor{color2bg}{HTML}{EDB120}
\definecolor{color3bg}{HTML}{0072BD}
\definecolor{color4bg}{HTML}{A2142F}

\providecommand{\norm}[1]{\lVert#1\rVert}
\providecommand{\SRG}[1]{\operatorname{SRG}(#1)}

\def\tb#1{{\color{black}#1}} 
\def\tbr#1{{{\color{red}#1}}} 

\tikzset{sin v source/.style={
  circle,
  draw,
  append after command={
    \pgfextra{
    \draw
      ($(\tikzlastnode.center)!0.5!(\tikzlastnode.west)$)
       arc[start angle=180,end angle=0,radius=0.425ex] 
      (\tikzlastnode.center)
       arc[start angle=180,end angle=360,radius=0.425ex]
      ($(\tikzlastnode.center)!0.5!(\tikzlastnode.east)$) 
    ;
    }
  },
  scale=1.5,
 }
}

\begin{document}

\title{Stability Analysis of Power-Electronics-Dominated Grids Using Scaled Relative Graphs}

\author{Eder~Baron-Prada, Adolfo Anta and~Florian Dörfler
\thanks{Eder Baron is with the Austrian Institute of Technology, 1210 Vienna, Austria, and also with the Automatic Control Laboratory, ETH Zurich, 8092 Z\"urich, Switzerland. (e-mail: ebaron@ethz.ch), 
Adolfo Anta is with the Austrian Institute of Technology, Vienna 1210, Austria (e-mail: adolfo.anta@ait.ac.at), 
Florian Dorfler is with the Automatic Control Laboratory, ETH Zürich, Zürich 8092, Switzerland (e-mail: doerfler@ethz.ch).}
}

\markboth{IEEE TPS}%
{Shell \MakeLowercase{\textit{et al.}}: Bare Demo of IEEEtran.cls for IEEE Journals}
\maketitle
\begin{abstract}

This paper presents a novel approach to stability analysis for grid-connected converters utilizing Scaled Relative Graphs (SRG). Our method effectively decouples grid and converter dynamics, thereby establishing a comprehensive and efficient framework for evaluating closed-loop stability. Our analysis accommodates both linear and non-linear loads, enhancing its practical applicability. Furthermore, we demonstrate that our stability assessment remains unaffected by angular variations resulting from dq-frame transformations, significantly increasing the method's robustness and versatility. The effectiveness of our approach is validated in several simulation case studies, which illustrate its broad applicability in modern power systems.
\end{abstract}


\IEEEpeerreviewmaketitle

\section{Introduction}
The integration of renewable energy sources is fundamentally reshaping modern power grids. Power electronic converters are rapidly replacing traditional synchronous generators and serve as the interface for renewable energy sources \cite{Farhangi2010_pathof}. While essential for enabling renewable energy systems, these converters introduce complex dynamics due to internal filters such as phase-locked loops, nested control loops, and intricate grid interactions. These factors present significant and unprecedented stability challenges \cite {dorfler2023,Huang2024_Howmany,Cigre2024}. Unlike synchronous generators, which inherently synchronize and provide grid-stabilizing inertia, converters lack such properties, making stability analysis and control more difficult in converter-dominated systems. 

\tb{There are two main approaches to analyze the small-signal stability of power systems. On one hand, state-space methods operate in the time domain and require full knowledge of the system matrices, making them suitable for detailed model-based analysis. Nonetheless, these models are not commonly accessible, in part because manufacturers maintain confidentiality over their control algorithms. \cite{bahrani2024grid}. On the other hand, impedance-based methods are formulated in the frequency domain, allowing stability and interaction studies using measured or simulated frequency responses obtained from black-box models, which are commonly shared by manufacturers}. {This makes them particularly compatible with data-driven identification and experimental validation.} 

Traditional \tb{impedance-based} small-signal stability analysis techniques (such as Nyquist-based methods) encounter limitations when applied to converter-dominated grids. These approaches aggregate the interconnected system dynamics, potentially obscuring frequency-specific instabilities \cite{Tao2022,Fan2020_problemsAdmittance}. Although impedance-based methods can provide insights for single-converter systems, they scale poorly in large, decentralized networks\cite{Fan2020_problemsAdmittance}. Recently, decentralized criteria like passivity, small-gain, and small-phase theorems have gained attention for their scalability, given that they can decouple converter and grid dynamics; however, they are often conservative and require strong assumptions \cite{Wang2024,huang2024gain,Wang2024Limitations}.  Moreover, in current power systems where nonlinearities are unavoidable, the limitations of these criteria become even more pronounced.

To address these challenges, we adopt the Scaled Relative Graphs (SRG) framework\cite{Ryu_2021,ryu2022large}, a novel methodology for stability analysis of interconnected systems \cite{Baron2025SRG,Baron2025SGP,Chaffey_2023,Chaffey2021,krebbekx2025}. When applied to power systems dominated by converters, SRG offers a significant advancement by allowing frequency-specific analysis and enabling a clear decoupling of grid and converter dynamics \cite{Baron2025SRG}. By separating grid and converter dynamics, SRG provides clearer insights into their individual behaviors and their interactions, facilitating targeted interventions and design improvements in power-electronic systems. This modularity enhances both analytical precision and practical usability.

Our contributions are as follows: we adapt a SRG-based stability theorem from \cite{Baron2025SRG} for power system analysis, creating a customized stability evaluation framework. We broaden the stability theorem's scope to include nonlinear loads, specifically, constant power loads (CPLs) not previously addressed through frequency-domain analysis. To achieve this, we develop an over-approximation technique for the SRG applicable to CPLs, which enables us to integrate CPLs directly into closed-loop stability assessments without requiring linearization. This methodology maintains the natural nonlinear characteristics of CPLs, resulting in enhanced accuracy of stability assessment.  Moreover, we establish that our stability results remain unaffected by angular changes arising from $dq$-frame coordinate transformations, which is critical for power systems with multiple converters modeled in distinct local and global $dq$-frames\cite{gong2018impact}.  

We validate our approach through a series of case studies, including single grid-connected converter and system-level stability assessment. For grid-following (GFL) converters, our SRG-based analysis corroborates existing results: higher phase-locked loop (PLL) cutoff frequencies ($f_{\text{pll}}$) require stronger grid conditions to maintain stability, consistent with~\cite{Li_duality_2022,Huang2020}. Leveraging the SRG framework's compatibility with nonlinear loads, we further perform stability assessments of grid-forming (GFM) converters connected to CPLs. Finally, we demonstrate the method's practical applicability through stability analysis of IEEE 14-bus and \tb{IEEE 57-bus systems }with high renewable penetration.
The SRG framework offers multiple advantages over conventional methods, outperforming passivity-based and mixed small-gain and phase techniques in certifying stability \cite{Baron2025SGP}.

This paper is organized as follows: Section \ref{sec:model} presents the modeling of converter and grid dynamics. Section \ref{sec:conditions_linear} develops stability certification for grid-connected converters with linear loads using SRGs, while Section \ref{sec:conditions_nonlinear} extends this framework to nonlinear loads. Section \ref{sec:system_analysis} uses the SRG-based theorem to analyze the stability of a system, followed by a comparison of stability analysis methods in Section \ref{sec:comparison}. Finally, Section \ref{sec:conclusions} concludes the work and discusses future research directions.

\section{Modeling Dynamics of Converter and Grids} \label{sec:model}
 
This section models the converter and grid using the admittance approach, linking Point of Common Coupling (PCC) terminal voltages and currents, and presents the converter-grid feedback loop.
\subsection{Admittance Modeling of Grid-Connected Converters}
Consider a converter connected to an AC system, as in Fig.~ \ref{fig:setup_single}, wherein the DC side is connected to a DC voltage source, $v_{dc,g}$ with an internal impedance $z_{dc,g}$. On the AC side, there is an LC filter that connects the converter to a balanced AC grid. At the PCC, a constant admittance load $y_{l}(s)$, a constant power load $y_{cp}$, and an infinite current bus $i_{ac,g}$ are connected. The infinite bus is considered in steady-state and therefore has no implication in small-signal stability. In Subsection \ref{sec:system_analysis}, we extend our analysis to scenarios where the converter at the PCC interfaces with a broader system.

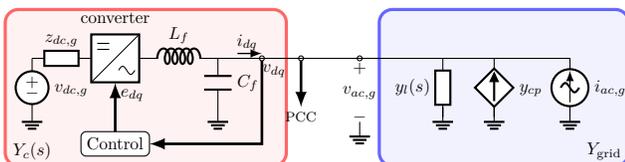
\begin{figure}[hbt]
    \centering
    \begin{circuitikz}[scale=0.65, every node/.style={transform shape}]
\filldraw[color=red!60, fill=red!5, very thick,rounded corners=5] (0.75,5) rectangle (6.5,1.8);
\filldraw[color=blue!60, fill=blue!5, very thick,rounded corners=5] (8.4,5) rectangle (13.5,1.8);
\filldraw[color=black, fill=white,rounded corners=2] (2.3,2.5) rectangle (3.7,2);
\node at (1.3,2.1) {$Y_c(s)$};
\node at (13,2.1) {$Y_{\text{grid}}$};
\node at (6.8,2.8) { \footnotesize PCC};
\node at (3,2.25) {Control};
\draw[fill = white] (6,4) circle [radius=0.05];
\draw[-latex, line width = 1 pt] (6,3.95) -- (6,2.25)--(3.7,2.25);
\draw[-latex, line width = 1 pt] (3,2.5) --  (3,3.5);
\draw[fill = white] (6.8,4) circle [radius=0.05];
\draw[-latex, line width = 1 pt] (6.8,3.95) -- (6.8,3);
\node at (5.7,4.33) {$i_{dq}$};
\node at (6.25,3.75) {$v_{dq}$};
\node at (8,3.3) {$v_{ac,g}$};
\node at (8,3.8) {$+$};
\node at (8,2.8) {$-$};
\node at (3.35,3.25) {$e_{dq}$};
\draw[-latex, line width = 0.5 pt] (5.5,4.1) --  (6,4.1);
\draw[fill = white] (8,4) circle [radius=0.05];

\ctikzset{capacitors/height=.4}
\draw
(1.3,2.8) to (1.3,2.7) node[tlground]{}
(1.3,2.8) to[american voltage source, l_=$v_{dc,g}$,sources/scale=0.8,invert, fill=white] (1.3,4) 
(1.3,4)to[twoport,l=$z_{dc,g}$,bipoles/twoport/width=0.5,bipoles/twoport/height=0.22, fill=white] (2.5,4)
(2.5,4) to [sdcac,fill=white, l=converter] (3.5,4)
(3.5,4) to [L, l=$L_{f}$] (5.1,4)
(5.1,4) to [C, l=$C_{f}$] (5.1,3)
(5.1,3) to (5.1,2.7) node[tlground]{}
(5.1,4) to (9.8,4)
(9.7,4) to [twoport,a=$y_{l}(s)$,bipoles/twoport/width=0.5,bipoles/twoport/height=0.22, fill=white] (9.7,2.8)
(9.8,4) to (12.3,4)
(9.7,3) to (9.7,2.7) node[tlground]{}
(12.3,4) to[sI, l=$i_{ac,g}$,sources/scale=0.9, fill=white] (12.3,2.8)
(12.3,3) to (12.3,2.7) node[tlground]{} 
(10.75,4) to[american controlled current source, l=$y_{cp}$,csources/scale=0.8,invert, fill=white] (10.75,2.8)
(10.75,3) to (10.75,2.7) node[tlground]{}
(8,2.7) to (8,2.4) node[tlground]{};
\draw [arrows = {-latex[scale=10]}]   (12.3,3.1)--(12.3,3.75)  ;
\end{circuitikz}
    \caption{A converter connected to a grid composed of an infinite bus in parallel to constant impedance and constant power load. }
    \label{fig:setup_single}
\end{figure} 

Depending on the objective of the converter connection, two control options are considered: GFL and GFM control. In the GFL case, the main objective is to exchange active and reactive power with an AC grid, i.e., acting as a power source. The main objective of the GFM is to regulate voltage and frequency at the PCC. In both cases, it is possible to model the converter small-signal dynamical behavior in the frequency domain via an admittance matrix, $Y_c(s)$. The admittance of the converter includes the LC filter, DC source, and the controller, as shown in Fig.~\ref{fig:setup_single}. This admittance model helps to capture the converter’s interaction with the AC grid, including its synchronization and control behaviors. The admittance model can be built in several coordinate systems, such as $\alpha\beta$, $dq$, or $abc$, depending on the assumptions on the grid\cite{Harnefors2007}. The admittance $Y_c(s)$ is typically built using a local $dq$-frame, and later related to a global $dq$-frame using rotation matrices:
\begin{align}
- \begin{bmatrix}
i_{d} \\
i_{q}
\end{bmatrix}
=  \underbrace{J(\theta) Y_{c}(s) J(-\theta)}_{\mathbf{\Tilde{Y}_{c}(s)}}
\begin{bmatrix}
v_{d} \\
v_{q}
\end{bmatrix} \label{eqn:Yc},
\end{align}
where $Y_{c}(s)$ is a $2\times2$ transfer function matrix that describes the converter’s dynamics using local per-unit calculations, $i~=~[i_d, i_q]^\top $ is the converter’s output current and $v=[v_d, v_q]^\top $ the input voltage, both in this global $dq$-frame.  The steady-state angle difference between the network reference voltage ($v_{ac,g}$), and the converter voltage coordinate  is represented by a constant $\theta$, where the rotation is given by 
\begin{align*}
    J(\theta) =
\begin{bmatrix}
\cos \theta & -\sin \theta \\
\sin \theta & \cos \theta
\end{bmatrix}.
\end{align*}

This global $dq$ coordinate system ensures uniformity in modeling converter-connected systems \cite{huang2024gain,Huang2024_Howmany}. 
\begin{remark}
\tb{The $dq$ transformation assumes balanced sinusoidal conditions. Under unbalance, positive and negative sequences couple, producing $2\omega$ oscillations and loss of decoupling. In such cases, a harmonic state-space (HSS) formulation is more suitable as it preserves sequence interactions.}
\end{remark}

\subsection{{Linear Load, Constant Power Load} and Grid Modeling } \label{subsec:modeling}
This section examines load modeling through two elements: linear admittance representation and CPL modeling. The discussion concludes with an examination of grid model representation and a simplified form (the short-circuit ratio) which serves as a measure of grid strength.

\subsubsection{{Linear Load} Admittance}
The admittance {transfer function} captures the linearized behavior of a load around a steady state as
\begin{align}
        y_{l}(s) = \begin{bmatrix}
        y_{dd}(s) & y_{dq}(s) \\
        y_{qd}(s) & y_{qq}(s)
    \end{bmatrix}. \label{eqn:Zgrid_estimated}
\end{align}
Note that constant admittance loads are linear and can be represented as \eqref{eqn:Zgrid_estimated} without losing generality. 

\subsubsection{Constant power loads}
A CPL is a nonlinear load that maintains a fixed power consumption regardless of voltage and frequency variations, adjusting its current to ensure the specified power demand is met~\cite{dorfler2023}. The relationship between $i$ and $v$ for a CPL, i.e., $i=y_{cp}(v)v$, can be described in local $dq$ coordinates as:
\begin{align} 
     \begin{bmatrix} i_d \\ i_q \end{bmatrix} = \dfrac{1}{\|v\|_2^2} \begin{bmatrix} p_c & q_c \\ -q_c & p_c \end{bmatrix} \begin{bmatrix} v_d \\ v_q \end{bmatrix},
     \label{eqn:Zgrid_CPL} 
 \end{align}
 where \tb{$\|\cdot\|_2$ denotes the $\mathcal{L}_2$-norm},  $p_c$ and $q_c$ denote constant active and reactive power set points, respectively. For brevity, hereafter we denote $ y_{cp} $ without explicit dependence on $ v $. 

\subsubsection{Grid Model}
{Our grid model from Fig.~\ref{fig:setup_single} can be represented as the sum of the admittances: 
 \begin{align}
     Z^{-1}_{\text{grid}}=Y_{\text{grid}}= y_{l}(s)+y_{cp}(v),
     \label{eqn:Zgrid_nonlinear}
 \end{align}
where $Z_{\text{grid}}$ denotes the grid impedance. We relate the grid impedance to voltages and currents in the global $dq$-frame as
\begin{align}
v= Z_{\text{grid}}\;i.    \label{eqn:Zgrid_closedloop}
\end{align}}
\subsubsection{Short Circuit Ratio (SCR)} \label{subsubsec:SCR}
In grid scenarios where modeling becomes computationally intensive, the SCR is used as a simplified grid model and as an efficient tool for the preliminary stability assessment of grid-connected converters \cite{Cigre2024}.  The SCR is defined as the $\mathcal{L}_2$-norm of the {linearized} grid admittance normalized by the converter admittance at the fundamental frequency $\omega_0$, \cite{Standard_DClink}, 
\begin{align*}
       \operatorname{SCR}~:=~\dfrac{\|Y_{\text{grid}}(\textup{j}\omega_0)\|_2}{\|Y_{c}(\textup{j}\omega_0)\|_2},
\end{align*}
where $Y_{\text{grid}}(s)$ is the grid equivalent admittance and $Y_{{c}}(s)$ the converter admittance.
A high SCR indicates a stiff grid resistant to disturbances, while a low SCR suggests susceptibility to voltage instability. This frequency-invariant positive scalar~\cite{Green2024} serves as a stability metric by capturing grid robustness~\cite{Huang2020}. 

\begin{remark}[Availability of admittance models]
\tb{Converter and grid admittances can be obtained either from analytical small-signal models derived from control and filter dynamics~\cite{Chen2019,Cigre2024,Fan2020_problemsAdmittance} or through experimental and data-driven identification methods~\cite{Gong2021,Lyu2024,Fan2023}. In practice, the availability of such admittance data is not a limiting factor, as they are routinely provided by manufacturers or identified locally by system operators~\cite{spp_2025,neso_2024}. Impedance-based representations are therefore well established and widely used for grid modeling and stability assessment~\cite{Dorfler2011,Haberle2023,Cigre2024}.}
\end{remark}
\subsection{Feedback Loop between grid and converter}
We can now establish the feedback loop between the converter admittance and the grid impedance.  
Together, \eqref{eqn:Yc} and \eqref{eqn:Zgrid_closedloop} form the system closed-loop dynamics shown in Fig.~\ref{fig:decentralizedfb}. 
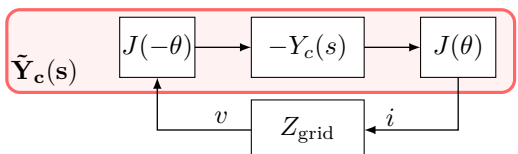
\begin{figure}[h]
\centering
\begin{tikzpicture}[scale=1, every node/.style={transform shape}]
\filldraw[color=red!60, fill=red!5, very thick,rounded corners=5] (6.5,4.6) rectangle (13.25,3.5);
\filldraw[color=black, fill=white] (9.75,4.5) rectangle (11.25,3.7);
\node at (10.5,4.1) {$-Y_c(s)$};
\filldraw[color=black, fill=white] (12,4.5) rectangle (13,3.7);
\node at (12.5,4.1) {$J(\theta)$};
\filldraw[color=black, fill=white] (8,4.5) rectangle (9,3.7);
\node at (8.5,4.1) {$J(-\theta)$};
\filldraw[color=black, fill=white] (9.75,3.4) rectangle (11.25,2.6);
\node at (10.5,3) {$Z_{\text{grid}}$}; 
\draw[-latex, line width = .5 pt] (9,4.1) -- (9.75,4.1);
\draw[-latex, line width = .5 pt] (11.25,4.1) -- (12,4.1);
\draw[-latex, line width = .5 pt] (12.5,3.7) -- (12.5,3) -- (11.25, 3 );
\draw[-latex, line width = .5 pt] (9.75,3) -- (8.5,3) -- (8.5,3.7) ;
\node at (9.35,3.1664) {$v$};
\node at (7,3.8) {$\mathbf{\Tilde{Y}_c(s)}$};
\node at (11.6,3.1664) {$i$};
\end{tikzpicture}
\caption{Closed-loop dynamics of a converter-grid system}
    \label{fig:decentralizedfb}
\end{figure} 
%
%

The converter’s dynamics and grid interaction form a closed-loop system that governs overall stability. However, certifying stability in such systems poses inherent challenges. Conventional methods such as the Generalized Nyquist Criterion (GNC) or Bode analysis are restricted to linear systems and lack explicit formulations to identify instability sources, often obscuring critical interactions between interconnected subsystems \cite{Fan2020_problemsAdmittance}. Other approaches, such as mixed small-gain and phase \cite{huang2024gain} or passivity techniques \cite{Wang2023_Passivity}, provide intuitive stability conditions but suffer from limited applicability and excessive conservatism. However, they can be applied in a decentralized manner. To overcome these limitations, Sections \ref{sec:conditions_linear} and \ref{sec:conditions_nonlinear} present an SRG-based method that provides stability conditions addressing nonlinear loads that decouples grid and converter dynamics. 

\section{Stability Certification for Grid connected converters with Linear Loads Using SRGs} \label{sec:conditions_linear}
This section presents an SRG method for analyzing stability in power systems. The definition and properties of SRGs are discussed, leading to Theorem \ref{thm:GFTLTI}, an SRG-based stability condition tailored for power systems. Applying this theorem, we determine the critical Short Circuit Ratio (cSCR) for a GFL with variable PLL bandwidth, thereby establishing the minimum grid strength threshold required for stable operation.
\subsection{Scaled Relative Graphs} 
First developed for convergence analysis in optimization algorithms \cite{Ryu_2021,ryu2022large}, SRGs have been extended to characterize both linear and nonlinear operators \cite{Chaffey_2023,Baron2025SRG}. Let $F \in \{\mathbb{R}, \mathbb{C}\}$ be the field of interest, and define the space $\mathcal{L}_2^n(F)$ as all square-integrable signals $u, y: \mathbb{R}_{\ge 0} \rightarrow F^n$ with inner product $\langle u, y \rangle := \int_0^\infty u(t)^* y(t) dt$ and norm $\|u\|_2 := \sqrt{\langle u, u \rangle}$. \tb{The angle between two nonzero vectors $z_1, z_2$ in the same Hilbert space is defined as $\angle(z_1, z_2):= \arccos\!\left( \tfrac{\Re\langle z_1,\, z_2 \rangle} {\|z_1\|_2\,\|z_2\|_2} \right)$}. For an operator $A$, the SRG captures input-output relationships as follows \cite{Ryu_2021}
\begin{align}
        \operatorname{SRG}(A) = \left\{ \frac{\|y_2 - y_1\|_2}{\|u_2 - u_1\|_2} \exp\left[\pm \textup{j} \angle(u_2 - u_1, y_2 - y_1)\right] \right\}
        \label{eqn:SRG_operator_nonlinear}
    \end{align}
where $u_1, u_2$ are a pair of inputs with resulting outputs $y_1~=~A(u_1)$ and $y_2~=~A(u_2)$. If A is a matrix, the SRG reduces to
 \begin{align}
        \operatorname{SRG}(A) = \left\{ \frac{\|Au\|_2}{\|u\|_2} \exp\left[\pm \textup{j} \angle(u,Au)\right] \right\}
        \label{eqn:SRG_operator}
    \end{align}
    for $\|u\|_2 \neq 0$ and $ \|A u\|_2\neq0$. Here, the magnitude ratio $\frac{\|Au\|_2}{\|u\|_2}$ quantifies gain, while the argument reflects phase shift. Calculating the SRG of a square LTI system can be done frequency-wise by applying \eqref{eqn:SRG_operator} to its transfer function matrix \cite{Baron2025SGP,Chen2025}. This work considers the $\mathcal{RH}_\infty$ space, i.e., rational, proper, and stable transfer functions, which are commonly used to represent admittances in power systems~\cite{huang2024gain,Huang2020}. For an admittance $Y(s)\in \mathcal{RH}_\infty$, at each frequency point $s=\textup{j}\omega$ with $\omega\in [0,\infty)$, the $\SRG{Y(\textup{j}\omega)}$ is given by  
 \begin{align}
        \operatorname{SRG}(Y(\textup{j}\omega)) = \left\{ \frac{\|Y(\textup{j}\omega)u\|_2}{\|u\|_2} \exp\left[\pm \textup{j} \angle(u,Y(\textup{j}\omega)u)\right] \right\},
        \label{eqn:SRG_operator_transferfunction}
    \end{align}
for $\|u\|_2 \neq 0$ and $ \|Y(\textup{j}\omega) u\|_2\neq0$. \tb{The SRG at each frequency captures all input–output behaviors of the system. Every point corresponds to a particular input direction $u$ and encodes the resulting gain and phase between input and output. Physically, this means that the SRG represents the full range of gain–phase responses attainable at that frequency, providing a geometric illustration of how the system reacts to all possible input directions. In the context of power systems, $\SRG{Y(\textup{j}\omega)}$ therefore represents the complete voltage-to-current gain–phase characteristics of the admittance at frequency $\omega$.} Thus, the $\SRG{Y(s)}$  can be seen as a three-dimensional set formed from  $\SRG{Y(\textup{j}\omega)}~\subset~\mathbb{C}$ for every $\omega\in\mathbb{R}$. {For SISO systems, the SRG reduces to the classical Nyquist plot \cite{Baron2025SRG,Chaffey_2023}, while for MIMO systems SRGs provide a more informative perspective.} The following three properties are essential for deriving SRG-based stability conditions\cite{ryu2022large}:
\begin{fact}[Inversion Property]\label{property:inversion}
    For any $A$, $\operatorname{SRG}(A^{-1})$ is obtained via the so-called Möbius inversion: 
    $\operatorname{SRG}(A^{-1})=\operatorname{SRG}(A)^{-1}=\{(z^{-1})^*\mid z\in \operatorname{SRG}(A)\}$.
\end{fact}
\begin{fact}[Chord Property]
An operator $A$ is said to satisfy the chord property if for every bounded $ z \in \operatorname{SRG}(A)$, the line segment $[z, z^*]$, defined as $z_1,z_2 \in \mathbb{C}$ as $[z_1,z_2] := \{ \beta z_1 + (1 - \beta)z_2 \mid \beta \in [0,1] \}$, is in $\operatorname{SRG}(A)$. 
\end{fact}
\begin{fact}[Sum of operators]\label{property:sum}
Let $A$ and $B$ be bounded operators. If either $\SRG A$ or $\SRG B$ meet the chord property. Then, $\SRG{A+B}\subseteq\SRG A + \SRG B$.
\end{fact}
We now present an adaptation of a SRG-based frequency-wise stability analysis tailored for power systems. This theorem certifies $\mathcal{L}_2$-stability \tb{(input-output stability)},  which means that for any bounded input signal, the corresponding output signal has finite energy, guaranteeing that the system response is not unbounded.

\begin{theorem}\label{thm:GFTLTI} \cite{chen2025softhardscaledrelative}
 Consider $\Tilde{Y}_c(s),Y_{\text{grid}}(s) \in \mathcal{RH}_\infty^{m\times m}$ in  closed-loop as in Fig. \ref{fig:decentralizedfb}, 
 If, $\forall s=\textup{j}\omega \text{ with }\; \omega \in [0, \infty)$,
\begin{align}
     \operatorname{SRG}({Y}_{\text{grid}}(s)) \cap -\tau\operatorname{SRG}(\Tilde{Y}_c(s)) = \emptyset,\quad\forall \tau \in (0,1],
     \label{eqn:SRGEquiLTI}
\end{align}
then the closed-loop system is $\mathcal{L}_{2}$ stable.
\end{theorem}
First, observe that $\operatorname{SRG}({Y}_{\text{grid}}(s)) = \operatorname{SRG}({Z}_{\text{grid}}^{-1}(s))$.  According to Condition~\eqref{eqn:SRGEquiLTI}, the closed-loop system is stable if $\operatorname{SRG}({Y}_{\textup{grid}}(s))$ and $-\tau\operatorname{SRG}(\Tilde{Y}_{c}(s))$ are disjoint at all frequencies. Notably, this stability criterion remains invariant regardless of whether $\tau$ scales $\operatorname{SRG}({Y}_{\textup{grid}}(s))$ or $-\operatorname{SRG}(\Tilde{Y}_{c}(s))$~\cite{Baron2025SRG}. \tb{The distance between two sets $A, B \subset \mathbb{C}$ is defined as $\operatorname{dist}(A,B)=\inf_{x\in A,y\in B}\norm{x-y}_{2}$. The non-zero distance} between SRGs serves as a frequency-domain stability margin, which is formally defined as follows~\cite{Baron2025decentralized}:
 \begin{definition}
     The \textit{SRG-based frequency-wise stability margin} of the feedback system between $\Tilde{Y}_c(s)$ and  ${Y}_{\text{grid}}(s)$ is defined for each $\omega\in[0,\infty)$ as
    \begin{align} 
    \rho(\omega):=\inf_{\tau\in(0,1]}\text{\tb{dist}}( \operatorname{SRG}({Y}_{\text{grid}}(s),-\tau\operatorname{SRG}(\Tilde{Y}_c(s))).
        \label{eqn:stability_margin}
    \end{align}
\end{definition}

The \textit{SRG-based stability margin} is defined at each frequency, providing insights into which frequency bands require improvement for stable operation. Also note that $\rho\in\mathbb{R}_{\geq0}$.  

\subsection{Case study 1: Determining the cSCR for a GFL converter} \label{subsec:scenario1}
{For didactic purposes, and as an initial example, we} analyze how the PLL bandwidth $f_{\text{pll}}$ affects GFL converter stability. The analysis utilizes the setup shown in Fig. \ref{fig:setup_single}, where the SCR is used as an equivalent grid model and the corresponding cSCR is determined using SRG. Simulations are performed using the Simplus toolbox\cite{Simplus}, and parameters are detailed in Appendix \ref{appendix:GFL_GFM}. A typical design of GFL controllers ensures that $Y_c(s)~\in~\mathcal{RH}_\infty^{2\times 2}$. Fig. \ref{fig:GFL_SCR} shows two 2D-projections of $-\SRG{Y_c(s)}$ into the complex plane with $f_{\text{pll}}=30$Hz and $f_{\text{pll}}=70$Hz with $\tau=1$. We only plot $-\SRG{Y_c(s)}$ for $\tau=1$, since as $\tau\to0$ $-\SRG{Y_c(s)}$ shrinks towards the origin. The $\SRG{\text{SCR}}$ is calculated using \eqref{eqn:SRG_operator}. Hence, its SRG corresponds to a point on the positive real axis. Fig.~\ref{fig:GFL_SCR} illustrates in orange the permissible SCR that ensures system stability. The critical SCR is defined by the intersection of $\SRG{\text{SCR}}$ with $-\tau\operatorname{SRG}(Y_c(s))$, marking the minimum grid strength required for stable operation, which is depicted as an orange point. For the GFL converter with $f_{\text{pll}}=30$Hz, the cSCR is found at 1.74 and occurs at 12.4 Hz, corresponding to where  $-\tau\SRG{Y_c(s)}$ intersects with $\SRG{\text{SCR}}$. A similar analysis is done for the GFL converter with $f_{\text{pll}}=70$ Hz. In this case, the cSCR should be higher than 2.81 to maintain stability.  This highlights the sensitivity of GFL converters to grid strength, with stability heavily dependent on  $f_{\text{pll}}$. Increasing $f_{\text{pll}}$ further influences this sensitivity, requiring a stronger grid for stability. These findings are consistent with reported results in \cite{Li_duality_2022,huang2024gain,Huang2020}.
\begin{figure}[h]
   \centering
\begin{subfigure}[b]{0.24\textwidth}
    \centering
    \includegraphics[width=1\textwidth]{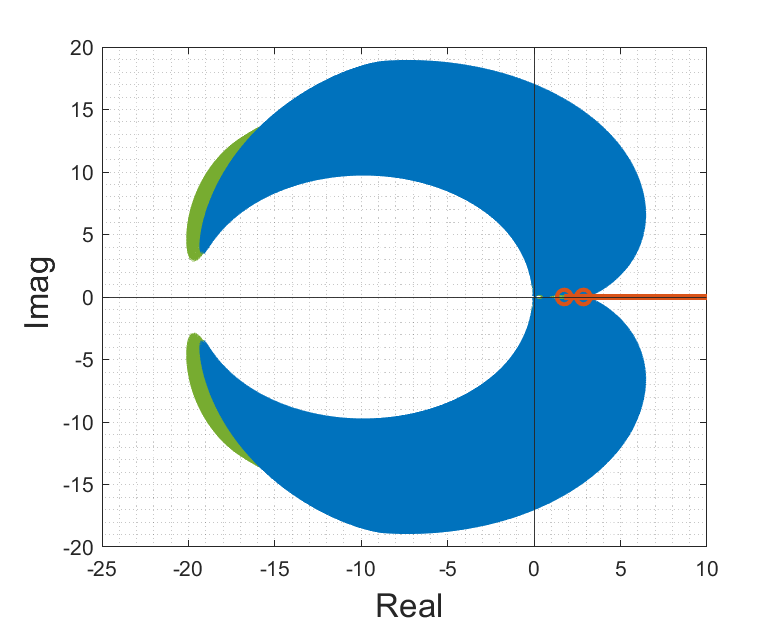}
    \caption{}
\end{subfigure}  
\begin{subfigure}[b]{0.24\textwidth}
    \centering
    \includegraphics[width=1\textwidth]{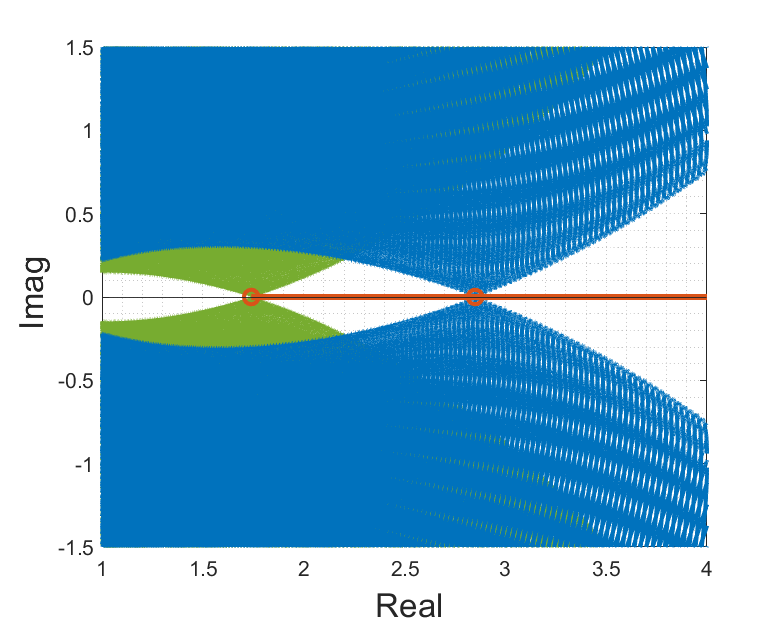}
    \caption{}
\end{subfigure}  
\begin{subfigure}[b]{0.24\textwidth}
    \centering
    \includegraphics[width=1\textwidth]{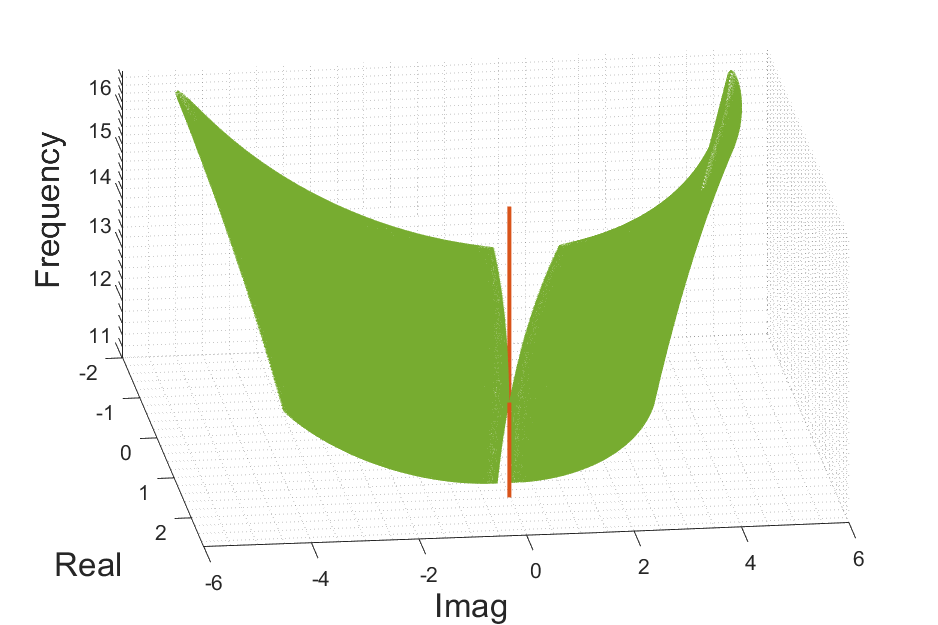}
    \caption{}
\end{subfigure}  
\begin{subfigure}[b]{0.24\textwidth}
    \centering
    \includegraphics[width=1\textwidth]{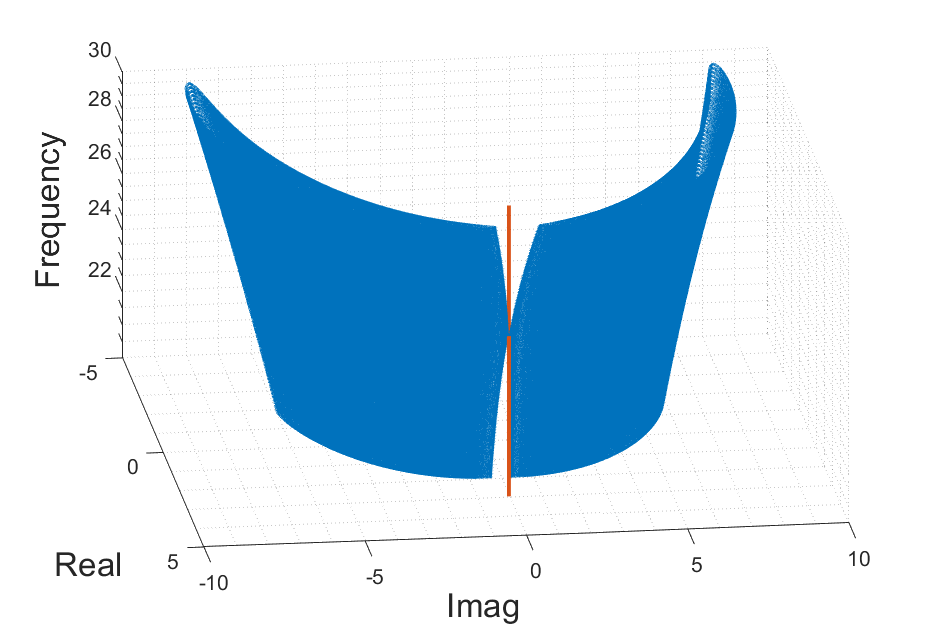}
    \caption{}
\end{subfigure}  
    \caption{\textbf{(a)} 2D SRG projection of the GFL with $f_{\text{pll}}=30$ Hz in green, with $f_{\text{pll}}=70$ Hz in blue, for $\omega\in [10^{-3},10^3]$Hz, and  all allowable SCR in orange. \textbf{(b)} Zoom on the 2D projection. \textbf{(c)} 3D SRG of GFL with $f_{\text{pll}}=30$ Hz, for $f \in [9,17]$Hz with cSCR in orange. \textbf{(d)} 3D SRG of GFL with $f_{\text{pll}}=70$ Hz, for $f \in [20,30]$Hz with cSCR in orange.}
    \label{fig:GFL_SCR}
\end{figure}

\subsubsection*{Time Domain Simulations Case Study 1}\label{subsec:Sim_GFL} 
 We consider three test cases with different SCR $= \{1.67, 1.74, 2\}$ where the frequency response is shown in Fig.~\ref{fig:Sims_GFLex1}. The system is stable for the GFL with $f_{\text{pll}}=30$ Hz if the SCR$\geq 1.74$. The critical case (SCR $= 1.74$) exhibits oscillations at approximately $12.4$ Hz, as depicted in Fig.~\ref{fig:GFL_SCR}. An SCR below $1.74$ leads to system instability, while values above the critical SCR (cSCR) result in stable grid operation. The same simulation is carried out for the GFL with $f_{\text{pll}}=70$Hz, with the SCR=$\{2.70,2.81,3.03\}$ showing analogous results. Under destabilizing SCR conditions, the observed oscillation frequencies match the frequencies at which the intersection between -$\tau\SRG{Y_c(s)}$ and $\SRG{\text{cSCR}}$  occurs for each converter.
\begin{figure}[ht]
    \centering
    \includegraphics[width=1\linewidth]{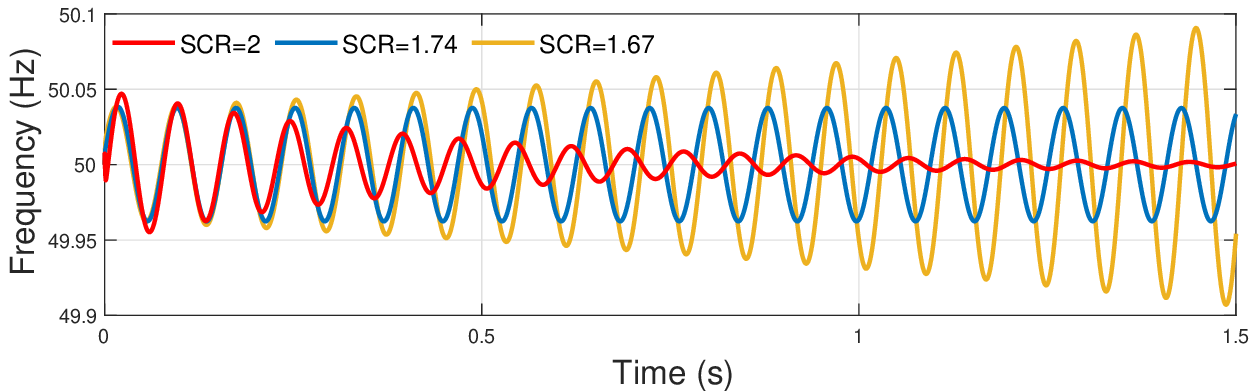}
    \caption{Frequency of the GFL converter with SCR=$\{1.67,1.74,2\}$ for $f_{\text{pll}}=30$Hz.}
    \label{fig:Sims_GFLex1}
    \includegraphics[width=1\linewidth]{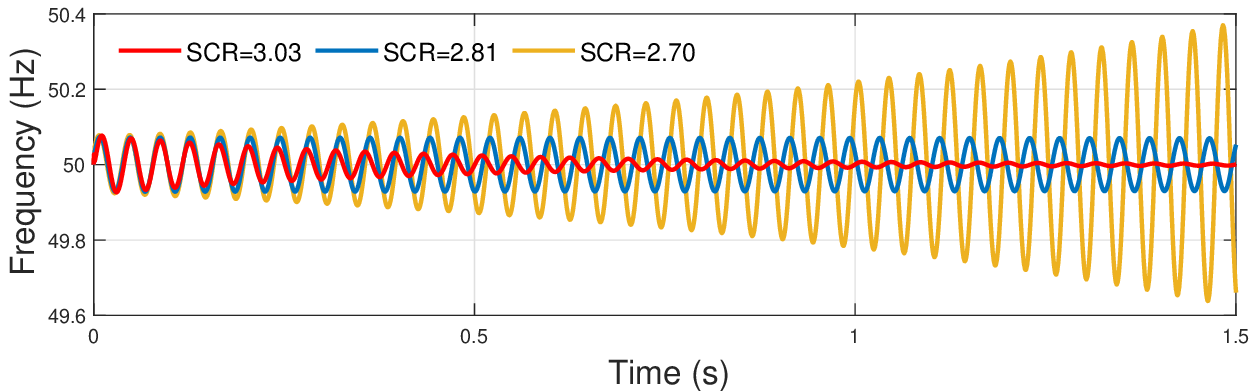}
    \caption{Frequency of the GFL converter with SCR=$\{2.70,2.81,3.03\}$ for $f_{\text{pll}}=70$Hz.}
    \label{fig:Sims_GFLex2}
\end{figure} 

\subsection{Relation to Traditional Gain and Phase Margins}
\tb{Classical gain and phase margins can be interpreted as special, scalar reductions of the geometric information contained in the SRG. For the gain margin, the standard MIMO extension uses the product of the largest singular value of each system in the feedback loop, which determines the maximum admissible increase in magnitude before the Nyquist plot reaches the critical point\cite{zhou1998,khalil2002}. In SRG terms, enclosing the SRG of each subsystem within a disk of radius equal to its maximum singular value produces this scalar projection: the minimal separation between these disks yields the classical gain margin. Because SRGs retain additional directional information that singular values discard, this disk-based reduction is more conservative than the full SRG separation test.}

\tb{The classical notion of phase margin, well defined for SISO systems, does not extend directly to the MIMO case. In such settings, a meaningful comparison arises with the numerical-range–based small-phase theorem~\cite{Wang2024} (see Section \ref{sec:comparison}), where a phase margin can be obtained by assuming each subsystem is sectorial over the frequency spectrum~\cite{Woolcock2023} and subtracting their maximal and minimal phase angles. This assumption is restrictive, as sectoriality rarely holds across all frequencies. In contrast, the SRG-based phase margin evaluates the SRG of each subsystem at every frequency, and for calculating the phase margin, an over-approximation by its maximum phase is made; stability is guaranteed when the sum of the two systems in feedback remains below~$\pi$~\cite{Baron2025SGP}.}
 
\tb{Overall, gain and phase margins correspond to scalar projections of the SRG geometry, whereas the SRG-based margin preserves the full coupled magnitude–phase structure of the input–output relation. This typically yields a larger certified stability region and reduces the conservatism inherent in traditional margin-based criteria.}

\section{SRG-based Stability Certification for Converters with Nonlinear Loads} \label{sec:conditions_nonlinear}
  
This section develops stability criteria for converters connected to CPLs, focusing on the feedback interconnection in Fig.~\ref{fig:decentralizedfb}. Our analysis builds upon Lemma~\ref{lemma:approx_SRG} that establishes a bound for the SRG of a CPL, a characterization of the linear and nonlinear energy CPL response derived in Lemma~\ref{lemma:CPL_eps_volterra}, and Theorem~\ref{thm:GFT_nonlinear} that provides stability conditions for converters connected to a CPL. 

\subsection{Frequency-wise SRG for CPL}
We formalize the CPL bound as follows:

\begin{lemma}\label{lemma:approx_SRG}
 Assume a CPL modeled by \eqref{eqn:Zgrid_CPL}, where $ p_c,~q_c~\in~\mathbb{R}_{\geq~0}$, and assume that there is $v_{\min}$ such that  $ \|v\|_2\geq v_{\min} \in\mathbb{R}_{>0}$. Then, $ \SRG{y_{cp}}~\subseteq~\SRG{\widehat{y_{cp}}}$, where $\SRG{\widehat{y_{cp}}}$ is defined as
\begin{align}
   \resizebox{0.85\hsize}{!}{$\SRG{\widehat{y_{cp}}}=\left\{re^{\mathrm{j}\alpha}|\forall \alpha \in [-\pi,\pi],\frac{\sigma_{\max}(M)}{v_{\min}^2} \geq r\geq 0\right\}$}, \label{eqn:bound_CPL}
\end{align}
where $M=\begin{bmatrix}
        p_c & q_c \\
        -q_c & p_c    
    \end{bmatrix}$.
\end{lemma}
\begin{proof}
    The proof can be found in the Appendix \ref{appendix:approximation}. 
\end{proof}

\tb{Lemma \ref{lemma:approx_SRG} establishes a} bound for $\SRG{y_{cp}}$ using a closed disk, \tb{it describes the behavior of the constant power load under normal operation, independently of the point of operation.} The bound reveals two relationships that impact system stability: first, lower values of $v_{\min}$ result in a larger $\SRG{\widehat{y_{cp}}}$, meaning that, as the minimum voltage decreases, the bounded SRG region expands. Second, both $p_c$ and $q_c$ exhibit a direct correlation with $\SRG{\widehat{y_{cp}}}$ size, where increases in either one lead to a larger $\SRG{\widehat{y_{cp}}}$, with both conditions negatively affecting system stability.

\tb{To enable a  frequency-wise analysis, we interpret the CPL as a bounded nonlinear operator that is, in practice, \emph{approximately frequency-preserving}.  Specifically, its input–output behavior is largely governed by a linear, frequency-preserving component, while the remaining nonlinear contribution is bounded and comparatively small. This viewpoint differs from a classical Taylor linearization, where the nonlinear map is replaced by a \textit{linear approximation} whose gain and phase depend on the operating point, and it is only valid in a small neighborhood of the operating point. By contrast, we work with a \textit{global over-approximation} obtained by decomposing the voltage as $v = v_0 + v_\delta$, where $v_0$ denotes an admissible operating voltage, with a norm of $V_0 := \|v_0\|_2 \geq v_{\min}$, and define the ripple ratio $\rho := \|v_\delta\|_\infty / V_0$. With this notation, the CPL current admits an exact decomposition into a frequency-preserving component and a bounded nonlinear remainder. The latter collects all harmonic terms, whose energy scales with the ripple~$\rho$ and is typically small under standard total harmonic distortion (THD) levels~\cite{Standard_Harmonics}. Since THD bounds the RMS energy of all non-fundamental components, the remainder remains small whenever THD lies within normal operating limits (typically $<5\%$, corresponding to $\rho<5\%$). The bound on the nonlinear remainder is formalized in Lemma \ref{lemma:CPL_eps_volterra}.}

\begin{lemma}
\label{lemma:CPL_eps_volterra}
\tb{Consider the CPL defined in \eqref{eqn:Zgrid_CPL}, and assume $\|v\|_2\ge v_{\min}$. Let $v=v_0+v_\delta$ with $V_0=\|v_0\|_2>0$ and ripple ratio $\rho=\|v_\delta\|_\infty/V_0$. Let $i_{\mathrm{har}}$ denote the nonlinear remainder of the CPL current. Then
\[
\|i_{\mathrm{har}}\|_2\ \le\ \frac{\sigma_{\max}(M)}{v_{\min}^2}\;
\frac{(\rho+1)(\rho+2)}{1-(2\rho+\rho^2)}\;\|v_\delta\|_2,
\]
and hence the CPL is $\varepsilon(\rho)$–frequency–preserving with
\[
\varepsilon(\rho)=
\frac{\sigma_{\max}(M)}{v_{\min}^2}\;
\frac{(\rho+1)(\rho+2)}{1-(2\rho+\rho^2)}.
\]}
\end{lemma}

\begin{proof}
    The proof can be found in the Appendix \ref{appendix:lemmaCPL}.
\end{proof}

\tb{The function $\varepsilon(\rho)$ bounds the $\mathcal L_2$–gain from the voltage ripple to the harmonic remainder of the CPL current. It increases monotonically with the ripple level $\rho$ but remains small for $\rho \le 5\%$, a range consistent with standard grid-code harmonic limits~\cite{Standard_Harmonics}. Lemma~\ref{lemma:CPL_eps_volterra} therefore shows that, when the voltage ripple is small, the nonlinear contribution is uniformly bounded and comparatively minor with respect to the linear component. In this regime, the CPL behaves predominantly as a frequency-preserving device, which allows its SRG to be evaluated on a per-frequency basis without introducing significant error. Empirical observations and numerical validation (see Subsection~\ref{subsec:Scenario2_nonlinear}) confirm that, under typical operating conditions ($\rho\le 5\%$), the nonlinear terms are indeed small, thus supporting the use of a per-frequency SRG-based stability analysis and the justification for Assumption \ref{assumption:nonharmonics}.} 

\begin{assumption}\label{assumption:nonharmonics}
\tb{The CPL exhibits a frequency-preserving response: for any input voltage $v(t)~=~\sum_{k=1}^N~\hat{v}_k~\cos(\omega_k t + \phi_k)$, the current responds as $i(t) = \sum_{k=1}^N f_k(\hat{v}_k) \cos\big(\omega_k t + h_k(\phi_k)\big),$ where $f_k(\cdot)$ and $h_k(\cdot)$ are independent of $\omega_k$.} 
\end{assumption} 

\begin{remark}[Frequency-wise CPL model]
\label{remark:frequency_CPL}
\tb{Lemma~\ref{lemma:CPL_eps_volterra} implies that the CPL admits a frequency-wise decomposition: a dominant frequency-preserving response together with an nonlinear term whose $\mathcal{L}_2$-gain is bounded by $\varepsilon(\rho)$. When the ripple is small, the nonlinear current is also small, and equipped with Assumption \ref{assumption:nonharmonics}, the SRG disk~\eqref{eqn:bound_CPL} applies uniformly at each frequency.}
\end{remark}

\begin{example}[SRG of linear load, CPL and their sum]
We compute the SRG for a linear load and a CPL as $\SRG{y_{cp} + y_{l}(s)}$. The following expressions define both loads:
\begin{align*}
    y_{l}(s) &= \begin{bmatrix}
        5+0.5s & -0.5 \\
        0.5 & 5+0.5s    
    \end{bmatrix}^{-1},
    y_{cp} &= \frac{1}{\|v\|_2^2} \begin{bmatrix}
        0.4 & 0.7 \\
        -0.7 & 0.4
    \end{bmatrix},
\end{align*}
with $v_{\min} = 0.9$.  $\SRG{y_{l}(s)}$ is computed using \eqref{eqn:SRG_operator_transferfunction}, while $\SRG{y_{cp}}$ is derived from \eqref{eqn:SRG_operator_nonlinear}. 

\begin{figure}[h]
    \centering
\begin{subfigure}[b]{0.155\textwidth}
    \centering
    \includegraphics[width=1\linewidth]{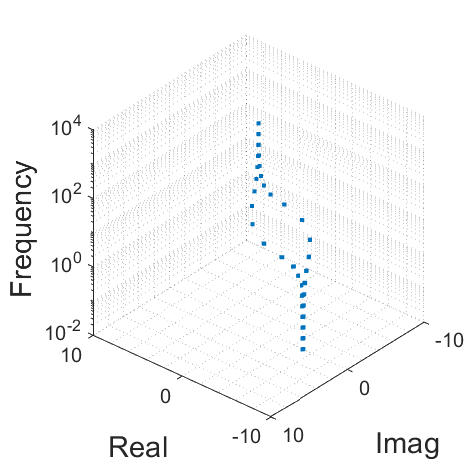}
    \caption{}
\end{subfigure}
\begin{subfigure}[b]{0.155\textwidth}
    \centering
    \includegraphics[width=1\linewidth]{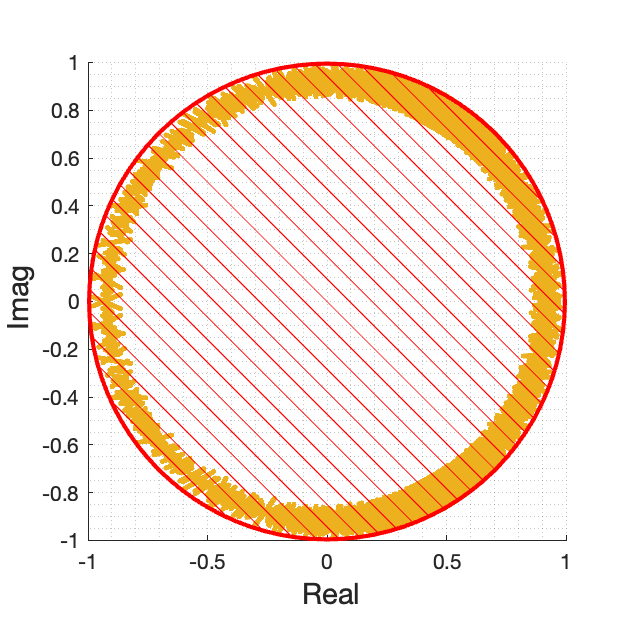}
    \caption{}
\end{subfigure}
\begin{subfigure}[b]{0.155\textwidth}
    \centering
    \includegraphics[width=1\linewidth]{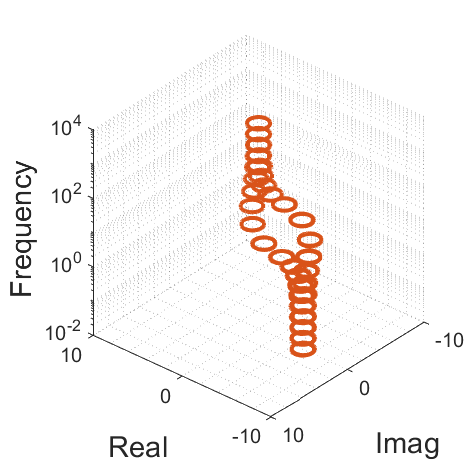}
    \caption{}
\end{subfigure}
\caption{(a) 3D $\SRG{y_{l}(s)}\forall \omega\in [10^{-2},10^4]$Hz in blue  (b) the exact 2D $\SRG{y_{cp}}$ in yellow and $\SRG{\widehat{y_{cp}}}$  as the hatched red disk. (c) $\SRG{\widehat{y_{cp}}}+\SRG{y_{l}(s)}$ in orange. }
    \label{fig:SRG_example}
\end{figure} 

First, Fig.~\ref{fig:SRG_example}(a) shows $\SRG{y_{l}(s)}$ over  $\omega \in [10^{-2}, 10^4]$~Hz. Next, Fig.~\ref{fig:SRG_example}(b) shows $\SRG{{y_{cp}}}$ and $\SRG{\widehat{y_{cp}}}$, which is calculated using Lemma~\ref{lemma:approx_SRG}. Notably, $\SRG{\widehat{y_{cp}}}$ remains constant across all frequencies aligned with Remark~\ref{remark:frequency_CPL}. Since $\SRG{\widehat{y_{cp}}}$ satisfies the chord property, $\SRG{y_{cp} + y_{l}(s)}$ can be approximated, using Property~\ref{property:sum} as $\SRG{\widehat{y_{cp}}} + \SRG{y_{l}(s)}$ depicted in Fig.~\ref{fig:SRG_example}(c). Note that we only plot the operator's border at some frequencies for better visibility.
\end{example} 
\subsection{Stability conditions for a converter connected to a grid with non-linear loads}
The following theorem constitutes our main contribution:
\begin{theorem}\label{thm:GFT_nonlinear} 
Suppose Assumption~\ref{assumption:nonharmonics} holds, and consider $Y_{c}(s), y_{l}(s) \in \mathcal{RH}_\infty^{2\times2}$ and $y_{cp} \in \mathcal{L}_{2}$ interconnected in closed loop as shown in Fig.~\ref{fig:decentralizedfb}. If for all $s = \textup{j}\omega$ with $\omega \in (0,\infty]$ and for all $\tau \in (0,1]$,
\begin{align}
    -\left(\operatorname{SRG}(y_{l}(s)) + \operatorname{SRG}(\widehat{y_{cp}})\right) \cap \tau \operatorname{SRG}(Y_{c}(s)) = \emptyset, \label{eqn:cor1}
\end{align}
then the closed-loop system is $\mathcal{L}_2$-stable.
\end{theorem}

\begin{proof}
    The proof can be found in the Appendix \ref{appendix:proofcor} 
\end{proof}

The main advantage of Theorem \ref{thm:GFT_nonlinear} is the possibility of analyzing the stability of the power system by comparing the $\operatorname{SRG}$ of the grid and the converter independently, providing insights into the stability limits of the converter or grid under study. Furthermore, this method allows the inclusion of non-linear loads into the small-signal stability assessment. \tb{With appropriate modifications, this approach can be extended to derive decentralized stability conditions\cite{Baron2025decentralized} (see also Remark \ref{rem:decentralized_setting} later on this paper).}

\begin{remark}[Alternative result]
The frequency-wise analysis in our approach requires Remark~\ref{remark:frequency_CPL} to obtain non-conservative results. One may alternatively apply the method from \cite[Thm 4]{Chaffey_2023}, which compares the sets $\tau  \operatorname{SRG}(Y_{c}(s))$ and $-(\operatorname{SRG}(y_{l}(s)~+~ y_{cp}))$ simultaneously across all frequencies. This alternative formulation leads to more conservative stability conditions than our frequency-wise method \cite{Baron2025SRG}.
\end{remark}
\begin{remark}[$dq$-{frames}]
Theorem \ref{thm:GFT_nonlinear} {implicitly} establishes that the stability criterion remains invariant regardless of the specific $dq$-frame employed by individual converters {as the SRG is unitarily invariant} (See Appendix \ref{appendix:proofcor}). This key property is not only applicable to $dq$-{frames} but also to any unitary transformation and, therefore, offers a significant advantage when assessing the stability of heterogeneous inverter-based systems, as it ensures a unified analytical framework without requiring converter-specific adjustments\cite{Shah2022_dqframes}. The practical benefits of this universality are demonstrated in Section \ref{sec:system_analysis}.
\end{remark} 
 
\subsection{Case Study 2:  GFM connected to non-linear grid model}\label{subsec:Scenario2_nonlinear}

In this case study, we analyze a GFM converter connected to three components: a linear impedance, a CPL, and a GFL converter connected to a constant DC current source. The setup is illustrated in Fig. \ref{fig:Nonlinearload_example} and the parameters and controllers are found in Appendix \ref{appendix:GFL_GFM} .

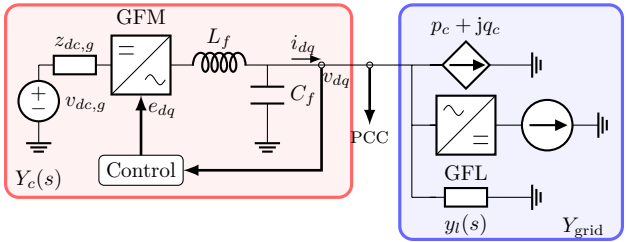
\begin{figure}[h]
\centering 
\begin{circuitikz}[scale=0.8, every node/.style={transform shape}]
\filldraw[color=red!60, fill=red!5, very thick,rounded corners=5] (0.75,5) rectangle (6.5,1.8);
\filldraw[color=blue!60, fill=blue!5, very thick,rounded corners=5] (7.3,5) rectangle (11,1.1);
\filldraw[color=black, fill=white,rounded corners=2] (2.3,2.5) rectangle (3.7,2);
\node at (1.3,2.1) {$Y_c(s)$};
\node at (10.35,1.3) {$Y_{\text{grid}}$};
\node at (6.8,2.8) { \footnotesize PCC};
\node at (3,2.25) {Control};
\draw[fill = white] (6,4) circle [radius=0.05];
\draw[-latex, line width = 1 pt] (6,3.95) -- (6,2.25)--(3.7,2.25);
\draw[-latex, line width = 1 pt] (3,2.5) --  (3,3.5);
\draw[fill = white] (6.8,4) circle [radius=0.05];
\draw[-latex, line width = 1 pt] (6.8,3.95) -- (6.8,3);
\node at (5.7,4.33) {$i_{dq}$};
\node at (6.25,3.75) {$v_{dq}$};
\node at (3.35,3.25) {$e_{dq}$};
\draw[-latex, line width = 0.5 pt] (5.5,4.1) --  (6,4.1);
\ctikzset{capacitors/height=.4}
\draw
(1.3,2.8) to (1.3,2.7) node[tlground]{}
(1.3,2.8) to[american voltage source, l_=$v_{dc,g}$,sources/scale=0.8,invert, fill=white] (1.3,4) 
(1.3,4)to[twoport,l=$z_{dc,g}$,bipoles/twoport/width=0.5,bipoles/twoport/height=0.22, fill=white] (2.5,4)
(2.5,4) to [sdcac,fill=white, l=GFM ] (3.5,4)
(3.5,4) to [L, l=$L_{f}$] (5.1,4)
(5.1,4) to [C, l=$C_{f}$] (5.1,3)
(5.1,3) to (5.1,2.7) node[tlground]{}
(5.1,4) to (8.2,4)
(7.8,4) to[american controlled current source, l=$p_c+\textup{j}q_c$,csources/scale=0.8, fill=white] (9,4)
(9,4) to (9.5,4) node[tlground,rotate=90]{}
(7.5,4) to (7.5,1.8)
(9,3) to [sdcac,fill=white, l=GFL] (7.8,3)
(9,3) to [american current source, csources/scale=0.8, fill=white] (10.5,3)
(10.5,3) to (10.6,3) node[tlground,rotate=90]{}
(7.5,3) to (8.2,3)
(9,1.8)to[twoport,l=$y_l(s)$,bipoles/twoport/width=0.5,bipoles/twoport/height=0.22, fill=white] (7.8,1.8)
(7.5,1.8) to (7.9,1.8)
(9,1.8) to (9.5,1.8) node[tlground,rotate=90]{};
\end{circuitikz}
  \caption{A GFM converter connected to a non-linear  $Y_{\text{grid}}$. }
  \label{fig:Nonlinearload_example}
  \end{figure}

The admittances of the GFM converter, the CPL, the GFL converter, and the linear impedance load are denoted as $Y_c(s)$, $y_{cp}$, $y_{gfl}(s)$, and $y_{l}(s)$, respectively. In this case study, we consider two cases with two different CPLs, CPL$_1$ with $p_c=0.1,q_c=0.1$ and CPL$_2$ with $p_c=0.56,\;q_c=0.1$ that are described by the admittances $y_{cp,1}$ and $y_{cp,2}$, respectively. The grid admittance $Y_{\text{grid}}$ can be expressed as the sum of the individual admittances $Y_{\text{grid,i}} = y_{cp,i} + y_{gfl}(s) + y_{l}(s)$ for $i\in\{1,2\}$.

 $\SRG{\widehat{y_{cp,1}}}$ is represented by the red hatched disk in Fig.~\ref{fig:SRG_GFM_Nonlinear}(a). The exact SRG—computed from input-output measurements—is shown in yellow for comparison. Fig. ~\ref{fig:SRG_GFM_Nonlinear}(b) shows the bound \tb{in red} for $y_{cp,2}$ and in magenta the exact SRG. Fig. ~\ref{fig:SRG_GFM_Nonlinear}(c) shows $\SRG{y_{l}(s)}$, whereas  Fig.~\ref{fig:SRG_GFM_Nonlinear}(d) shows $\SRG{y_{gfl}(s)}$, both computed using \eqref{eqn:SRG_operator_transferfunction}. Moreover, Fig. ~\ref{fig:SRG_GFM_Nonlinear}(e) depicts  $\SRG{Y_c(s)}$; note that $\SRG{Y_c(s)}$ spans over wider space in comparison to $\SRG{y_{gfl}(s)}$ and  $\SRG{y_{l}(s)}$. 

 It is also noticeable that $\SRG{Y_{\text{grid},1}}$ in Fig.~\ref{fig:SRG_GFM_Nonlinear}(f) resembles Fig.~\ref{fig:SRG_GFM_Nonlinear}(d) enlarged by $\SRG{\widehat{y_{cp,1}}}$. This is expected as both, $\SRG{\widehat{y_{cp}}}$ and $\SRG{y_{l}(s)}$ meet the chord property, and therefore $\SRG{Y_{\text{grid,i}}}=\SRG{\widehat{y_{cp,i}} + y_{gfl}(s) + y_{l}(s)}=\SRG{\widehat{y_{cp,i}} }+\SRG{y_{gfl}(s)}+\SRG{y_{l}(s)}$. Fig.~\ref{fig:SRG_GFM_Nonlinear}(f) shows that  $\tau\SRG{Y_{\text{grid},1}}$ does not intersect $-\SRG{Y_c(s)}$ and therefore by Theorem~\ref{thm:GFT_nonlinear}, the system is stable. In contrast, Fig.~\ref{fig:SRG_GFM_Nonlinear}(g) shows the $\SRG{Y_{\text{grid},2}}$ and $\SRG{Y_{c}}$ having multiple intersections, and therefore stability cannot be guaranteed. The SRG-based stability margin in Fig.~\ref{fig:SRG_GFM_Nonlinear}(h) provides a clearer insight into the separation between SRGs, revealing that the smallest distances (or intersections for CPL$_2$) occur below and above the system's nominal frequency ($\approx 0.1, 8$ and $10^3$ Hz). %

\begin{figure}[ht]
    \centering
\begin{subfigure}[b]{0.15\textwidth}
    \centering
    \includegraphics[width=0.92\linewidth]{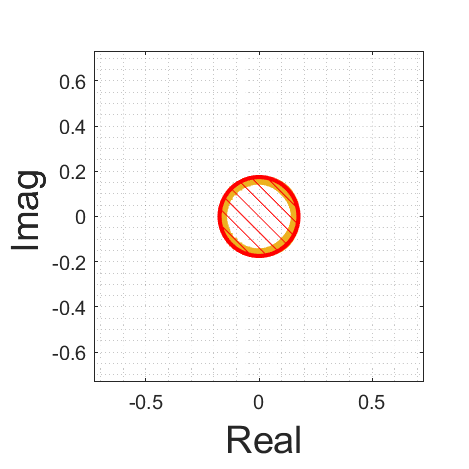}
\caption{}
\end{subfigure}
\begin{subfigure}[b]{0.15\textwidth}
    \centering
    \includegraphics[width=0.9\linewidth]{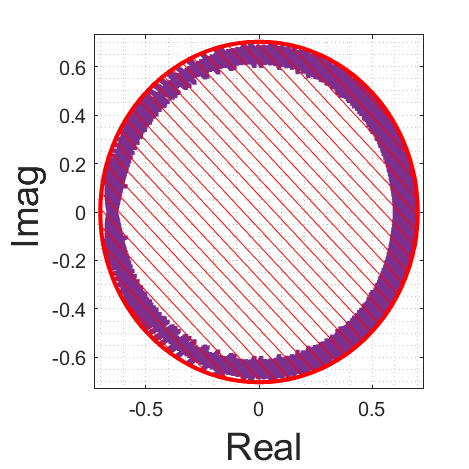}
\caption{}
\end{subfigure}
\begin{subfigure}[b]{0.15\textwidth}
    \centering
    \includegraphics[width=0.9\linewidth]{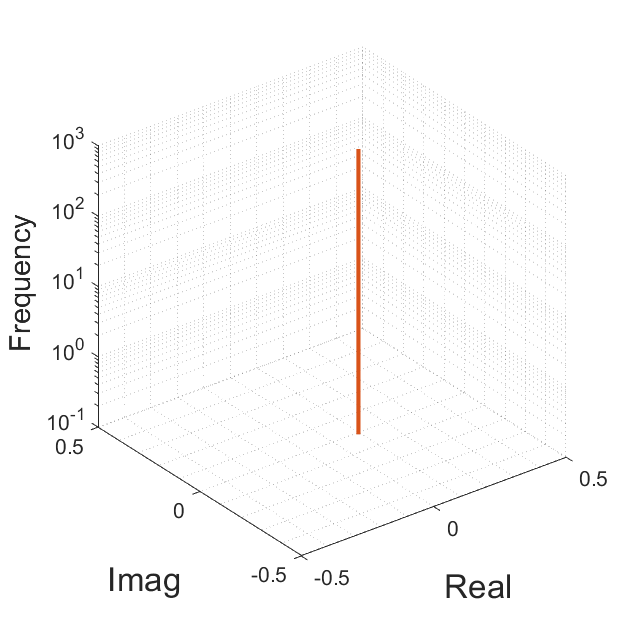}
    \caption{}
\end{subfigure}  
\begin{subfigure}[b]{0.24\textwidth}
    \centering
    \includegraphics[width=0.9\linewidth]{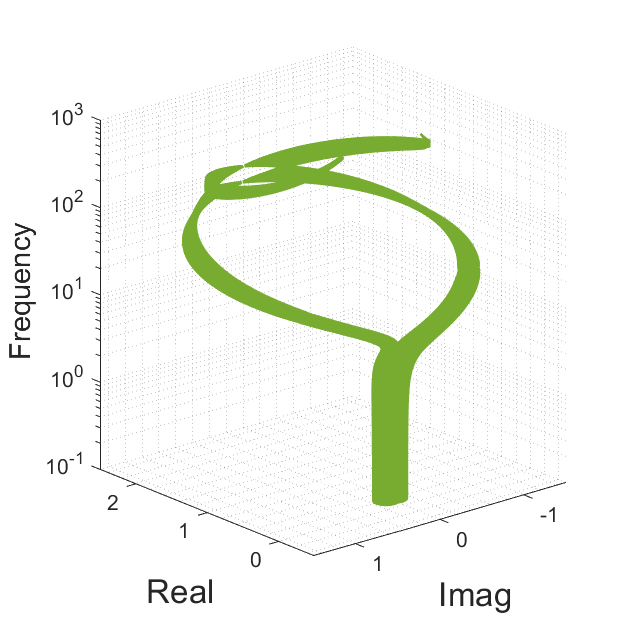}
    \caption{}
\end{subfigure}  
\begin{subfigure}[b]{0.24\textwidth}
    \centering
    \includegraphics[width=0.9\linewidth]{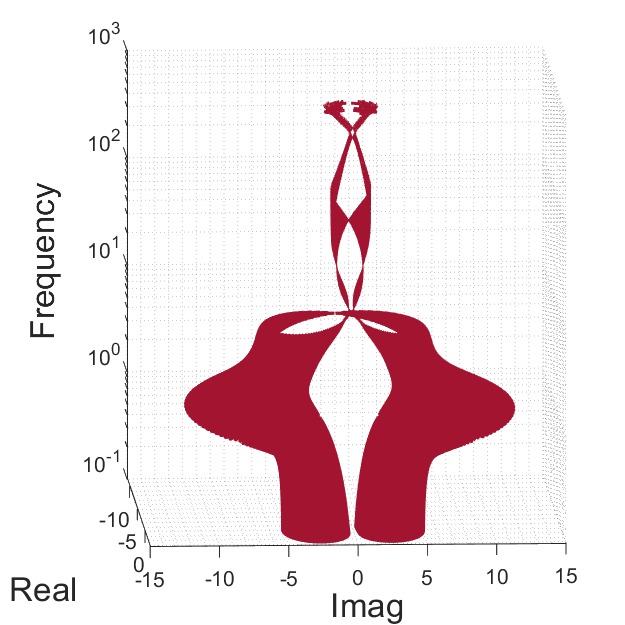}
    \caption{}
\end{subfigure}
\begin{subfigure}[b]{0.24\textwidth}
    \centering
    \includegraphics[width=0.9\linewidth]{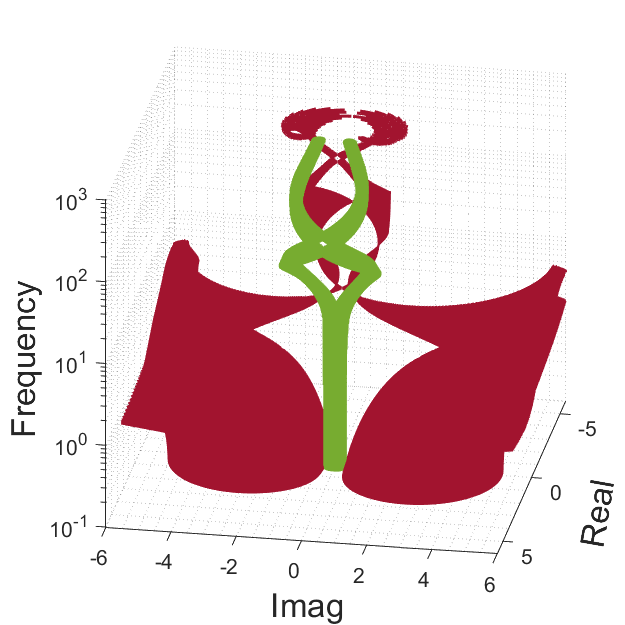}
    \caption{}
\end{subfigure}     
\begin{subfigure}[b]{0.24\textwidth}
    \centering 
    \includegraphics[width=0.9\linewidth]{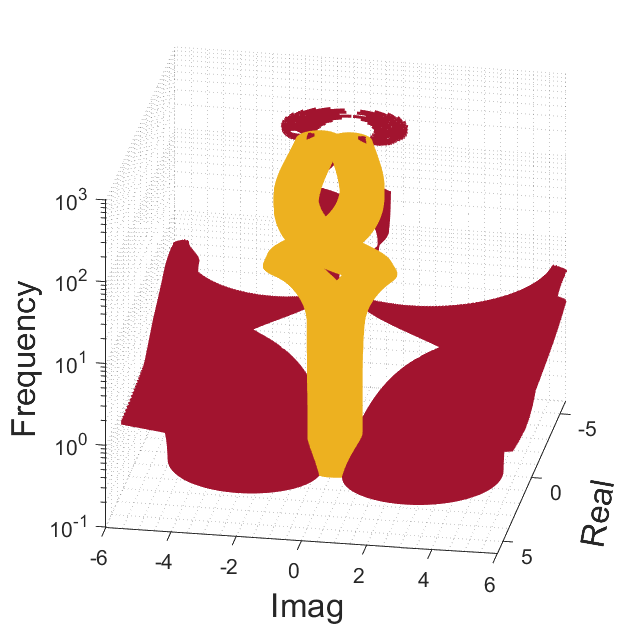}
\caption{}
\end{subfigure}
\begin{subfigure}[b]{0.5\textwidth}
    \centering
    \includegraphics[width=1\linewidth]{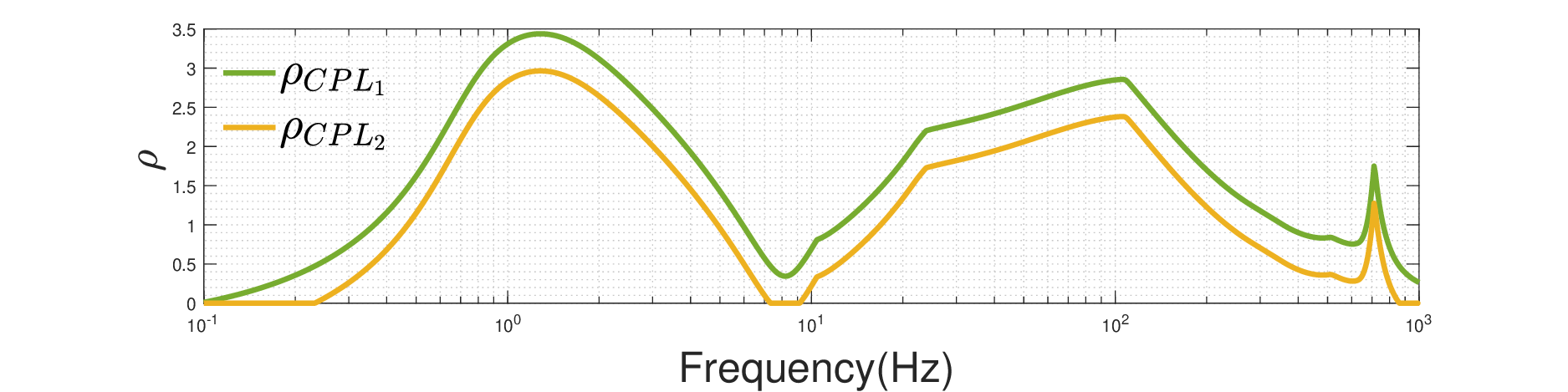}
    \caption{}
\end{subfigure}  
\caption{(a)~$\SRG{y_{cp,1}}$ in yellow and $\SRG{\widehat{y_{cp,1}}}$ red hatched disk.
(b)~$\SRG{y_{cp,2}}$ in magenta and $\SRG{\widehat{y_{cp,2}}}$ red hatched disk.
(c)~$\SRG{y_{l}(s)}$ in orange.
(d)~$\SRG{y_{gfl}(s)}$ in green. 
(e)~$\SRG{Y_c(s)}$ in red.
(f)~$\SRG{Y_c(s)}$ in red and -$\SRG{Y_{\text{grid,1}}}$ in green. 
(g)~$\SRG{Y_c(s)}$ in red and -$\SRG{Y_{\text{grid,2}}}$ in yellow. 
(h)~SRG-stability margin for both cases.}
    \label{fig:SRG_GFM_Nonlinear}
\end{figure} 
\subsubsection*{Time Domain Simulations Case Study 2} \label{subsec:Sim_Nonlinearloads}
We simulate the system in Fig.~\ref{fig:Nonlinearload_example}, starting with the GFL converter and linear load connected to the GFM converter. At $t = 1$~s, we introduce CPL$_1$. Fig.~\ref{fig:Sims_Nonlinear} shows the resulting frequency, reactive power, and active power dynamics. The GFL converter and $y_l(s)$ stabilize within 0.5~s, while the GFM converter exhibits the expected droop-like frequency response. The GFL converter initially oscillates near 50~Hz but settles within 0.6~s. Upon connecting CPL$_1$, the load instantly draws constant power, causing a grid imbalance that excites transient oscillations. These oscillations decay by $t = 1.8$~s, with the system frequency stabilizing at 49.93~Hz, slightly below the nominal 50~Hz but maintaining steady-state operation. Fig.~\ref{fig:Sims_Nonlinear_uns} shows identical behavior before $t = 1$~s. However, upon connecting CPL$_2$, the system dynamics diverge: while CPL$_2$ reaches steady-state power immediately, the GFM, GFL, and $y_l(s)$ exhibit oscillations for $\approx$1~s at around 8 Hz (same frequency band of intersections in Fig.~\ref{fig:SRG_GFM_Nonlinear}(g)). Beyond $t~>~2$~s, the system becomes unstable, as evidenced by growing deviations in frequency and power.  

\begin{figure}[ht]
    \centering
    \includegraphics[width=1\linewidth]{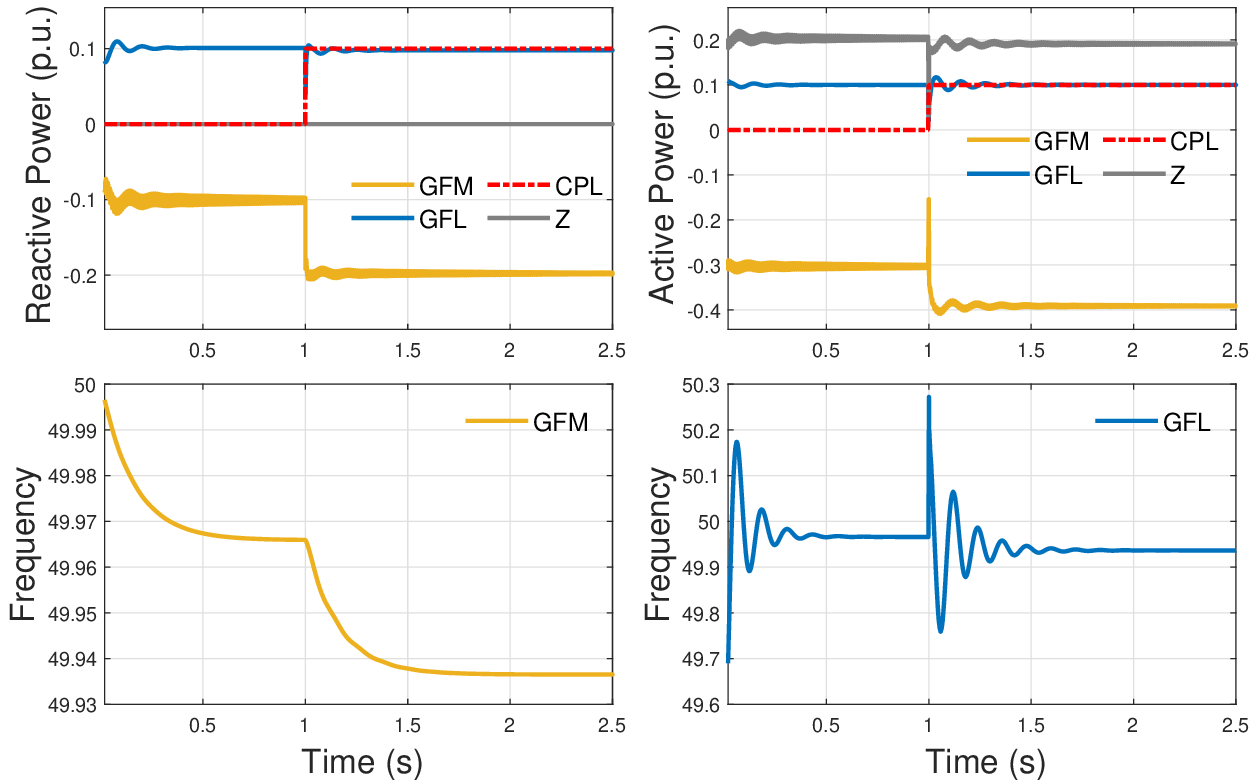}
    \caption{Reactive, active power, and frequency in the GFL and GFM depicted in Fig. \ref{fig:Nonlinearload_example} with CPL$_1$. }
    \label{fig:Sims_Nonlinear} 
    \end{figure} 
    \begin{figure}[ht]
    \centering
    \includegraphics[width=1\linewidth]{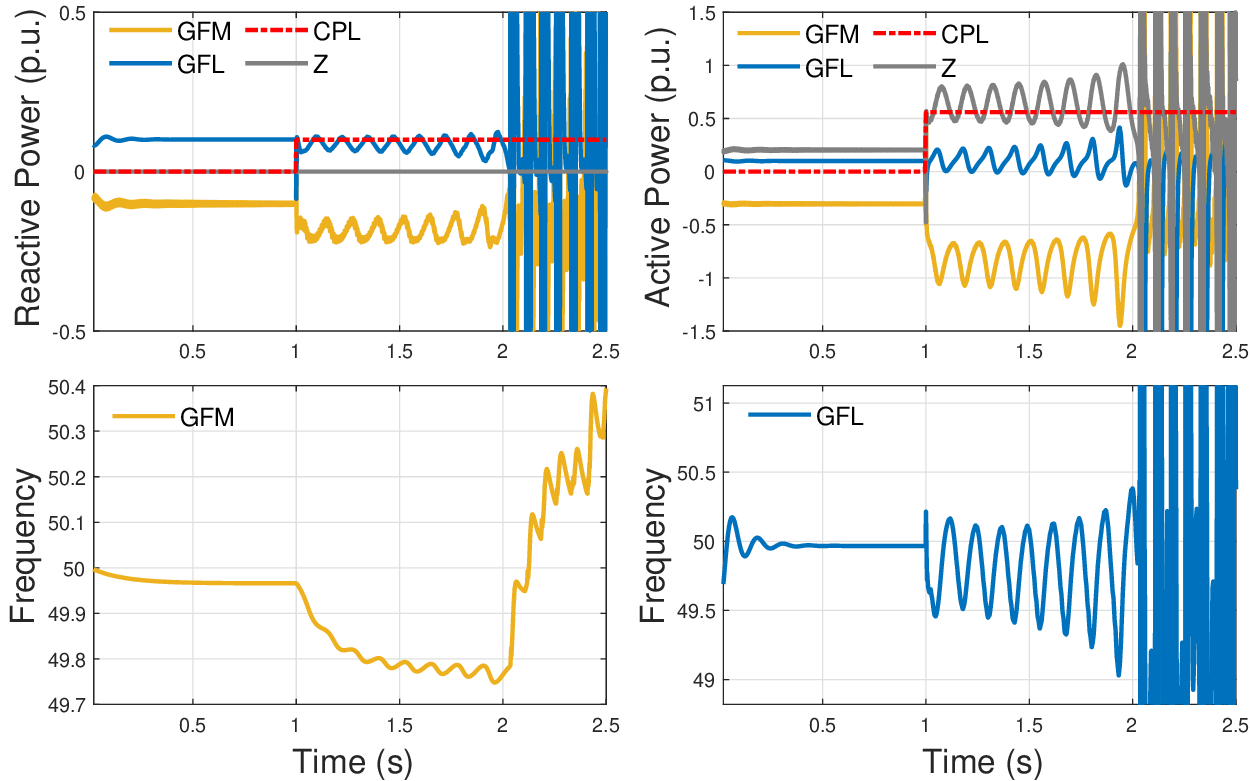}
    \caption{Reactive, active power, and frequency in the GFL and GFM depicted in Fig. \ref{fig:Nonlinearload_example} with CPL$_2$. }
    \label{fig:Sims_Nonlinear_uns}
\end{figure} 

\section{ System Stability Analysis }\label{sec:system_analysis}

This section investigates stability in the IEEE 14-bus system and the IEEE 57-bus system shown in Fig.~\ref{fig:14_node_IEEE} and Fig.~\ref{fig:57_node_IEEE}, with parameters detailed in Appendix \ref{appendix:system_analysis}.

\subsection{IEEE 14-node system}\label{subsec:14nodes}

The slack node is located in node 1 and modeled as a synchronous machine, while GFL converters are located in nodes 2, 3, 6, and 8 with admittances $Y_c^{N_2}(s),Y_c^{N_3}(s),Y_c^{N_6}(s)$ and $Y_c^{N_8}(s)$, respectively. {Note that all converters analyzed are GFL, but all have different $f_{\text{pll}}$}.  We analyze two scenarios: the first one comprises $Y_c^{N_8}(s)$ with $f_{i}=600$Hz and the second one $Y_c^{N_8}(s)$ with $f_{i}=100$Hz (where $f_{i}$ represents the frequency bandwidth of the current control loop) while the rest of the power system remains unchanged. 

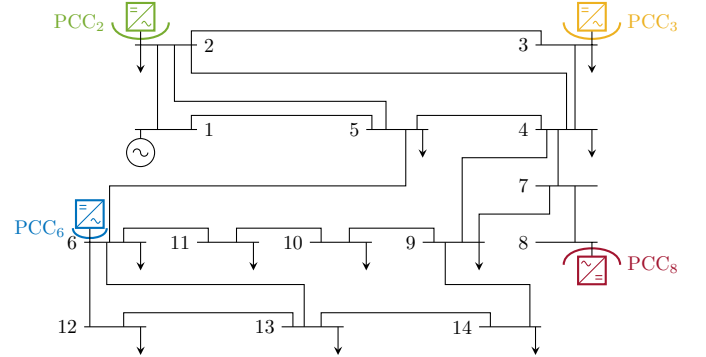
\begin{figure}[htb]
\centering
\begin{tikzpicture}[scale=0.75, every node/.style={transform shape}]
\draw[ thick,color=color1bg] (0.5,0.4) to [bend left=90] (-0.5,0.4)node[left]{PCC$_2$};
\draw[ thick,color=color2bg] (7.5,0.4) to [bend right=90] (8.5,0.4)node[right]{PCC$_3$};
\draw[ thick,color=color3bg] (-1.2,-3.25)node[left]{PCC$_6$} to [bend right=90] (-0.6,-3.25);
\draw[ thick,color=color4bg] (7.5,-3.9) to [bend left=90] (8.5,-3.9)node[right]{PCC$_8$};
\draw

(0,-1.5) coordinate(1-node) 
(1-node) --++ (0,-0.15) node[below,sin v source]{}
(1-node)--++(-0.1,0) 
(1-node)--++(1,0) node[right]{1} 
coordinate[pos=0.3](1-r) 
coordinate[pos=0.6](1-rr)
coordinate[pos=0.9](1-rrr)

(0,0)  coordinate(2-node) 
(2-node) --++ (0,0.3) node[above]{}
(2-node)--++(-0.1,0) 
(2-node)--++(1,0) node[right]{2} 
coordinate[pos=0.3](2-r) 
coordinate[pos=0.6](2-rr)
coordinate[pos=0.9](2-rrr)
;
\draw (0,0.5)node[sacdcshape,scale=-.5,color=color1bg,fill=white] (DC1){};
\draw
(8,0)  coordinate(3-node) 
(3-node) --++ (0,0.3) node[above]{}
(3-node)--++(0.1,0) 
(3-node)--++(-1,0) node[left]{3} 
coordinate[pos=0.3](3-r) 
coordinate[pos=0.6](3-rr)
coordinate[pos=0.9](3-rrr)
;
\draw (8,0.5)node[sacdcshape,scale=-.5,color=color2bg,fill=white] (DC2){};
\draw

(8,-1.5)  coordinate(4-node) 
(4-node)--++(0.1,0) 
(4-node)--++(-1,0) node[left]{4} 
coordinate[pos=0.3](4-r) 
coordinate[pos=0.6](4-rr)
coordinate[pos=0.9](4-rrr)


(5,-1.5)  coordinate(5-node) 
(5-node)--++(0.1,0) 
(5-node)--++(-1,0) node[left]{5} 
coordinate[pos=0.3](5-r) 
coordinate[pos=0.6](5-rr)
coordinate[pos=0.9](5-rrr)


(0,-3.5)  coordinate(6-node) 
(6-node)--++(0.1,0) 
(6-node)--++(-1,0) node[left]{6} 
coordinate[pos=0.3](6-r) 
coordinate[pos=0.6](6-rr)
coordinate[pos=0.9](6-rrr)
(6-rrr) --++ (0,0.3) node[above]{}
;
\draw (-0.9,-3)node[sacdcshape,scale=-0.5,color=color3bg,fill=white] (DC3){};
\draw

(8,-2.5)  coordinate(7-node) 
(7-node)--++(0.1,0) 
(7-node)--++(-1,0) node[left]{7} 
coordinate[pos=0.3](7-r) 
coordinate[pos=0.6](7-rr)
coordinate[pos=0.9](7-rrr)

(8,-3.5)  coordinate(8-node) 
(8-node)--++(0.1,0) 
(8-node)--++(-1,0) node[left]{8} 
coordinate[pos=0.3](8-r) 
coordinate[pos=0.6](8-rr)
coordinate[pos=0.9](8-rrr)
(8-node) --++ (0,-0.3) node[below]{}
;
\draw (8,-4)node[sacdcshape,scale=.5,color=color4bg,fill=white] (DC4){};
\draw
(6,-3.5)  coordinate(9-node) 
(9-node)--++(0.1,0) 
(9-node)--++(-1,0) node[left]{9} 
coordinate[pos=0.3](9-r) 
coordinate[pos=0.6](9-rr)
coordinate[pos=0.9](9-rrr)

(4,-3.5)  coordinate(10-node) 
(10-node)--++(0.1,0) 
(10-node)--++(-1,0) node[left]{10} 
coordinate[pos=0.3](10-r) 
coordinate[pos=0.6](10-rr)
coordinate[pos=0.9](10-rrr)


(2,-3.5)  coordinate(11-node) 
(11-node)--++(0.1,0) 
(11-node)--++(-1,0) node[left]{11} 
coordinate[pos=0.3](11-r) 
coordinate[pos=0.6](11-rr)
coordinate[pos=0.9](11-rrr)

(0,-5)  coordinate(12-node) 
(12-node)--++(0.1,0) 
(12-node)--++(-1,0) node[left]{12} 
coordinate[pos=0.3](12-r) 
coordinate[pos=0.6](12-rr)
coordinate[pos=0.9](12-rrr)

(3.5,-5)  coordinate(13-node) 
(13-node)--++(0.1,0) 
(13-node)--++(-1,0) node[left]{13} 
coordinate[pos=0.3](13-r) 
coordinate[pos=0.6](13-rr)
coordinate[pos=0.9](13-rrr)

(7,-5)  coordinate(14-node) 
(14-node)--++(0.1,0) 
(14-node)--++(-1,0) node[left]{14} 
coordinate[pos=0.3](14-r) 
coordinate[pos=0.6](14-rr)
coordinate[pos=0.9](14-rrr)
;
\draw[-stealth](2-node)--++(0,-0.5cm);
\draw[-stealth](3-node)--++(0,-0.5cm);
\draw[-stealth](4-node)--++(0,-0.5cm);
\draw[-stealth](5-node)--++(0,-0.5cm);
\draw[-stealth](6-node)--++(0,-0.5cm);
\draw[-stealth](9-node)--++(0,-0.5cm);
\draw[-stealth](10-node)--++(0,-0.5cm);
\draw[-stealth](11-node)--++(0,-0.5cm);
\draw[-stealth](12-node)--++(0,-0.5cm);
\draw[-stealth](13-node)--++(0,-0.5cm);
\draw[-stealth](14-node)--++(0,-0.5cm);


\draw
(1-r)--(2-r) 
(1-rrr)--++(0,0.25)--++(3.2,0)--++(0,-0.25) 
(2-rr)--++(0,-1)--++(3.75,0)--++(0,-0.5) 
(2-rrr)--++(0,-0.5)--++(6.65,0)--++(0,-1) 
(2-rrr)--++(0,0.25)--++(6.2,0)--++(0,-0.25) 
(3-r)--(4-r) 
(4-rrr)--++(0,0.25)--++(-2.2,0)--++(0,-0.25) 
(5-r)--++(0,-1)--++(-5.25,0)--++(0,-1) 
(6-r)--++(0,0.25)--++(1.5,0)--++(0,-0.25) 
(11-r)--++(0,0.25)--++(1.5,0)--++(0,-0.25) 
(10-r)--++(0,0.25)--++(1.5,0)--++(0,-0.25) 
(8-r)--(7-r) 
(4-rr)--(7-rr) 
(9-node)--++(0,0.5)--++(1.25,0)--++(0,0.5)
(9-r)--++(0,1.5)--++(1.5,0)--++(0,0.5)
(9-rr)--++(0,-0.75)--++(1.5,0)--++(0,-0.75) 
(12-r)--++(0,0.25)--++(3,0)--++(0,-0.25) 
(13-r)--++(0,0.25)--++(3,0)--++(0,-0.25) 
(6-rrr)--(12-rrr)%
(6-rr)--++(0,-0.75)--++(3.5,0)--++(0,-.75)
;
\end{tikzpicture}
\caption{Modified IEEE 14 node system.}
\label{fig:14_node_IEEE}
\end{figure}

We have two possible ways to build the feedback system in Fig. \ref{fig:decentralizedfb}, depending on the knowledge of the underlying grid connecting the converters. The first approach assumes no knowledge about the underlying grid. Therefore, we require the grid admittance transfer function seen from the PCC for the node under study, as given in \eqref{eqn:Zgrid_estimated}. For obtaining the grid admittance transfer functions from each PCC, we use the process proposed in \cite{Haberle2023}, injecting a signal into the system and deriving an admittance model that accurately matches the system's response as seen from nodes 2, 3, 6 and 8, which are denoted as $Y_{\text{grid}}^{N_2}(s),Y_{\text{grid}}^{N_3}(s),Y_{\text{grid}}^{N_6}(s)$ and $Y_{\text{grid}}^{N_8}(s)$, respectively. Importantly, this model incorporates the dynamics of other converters and the synchronous machine present in the system. We use Theorem~\ref{thm:GFT_nonlinear} for assessing the stability of the power system for each converter $i\in\{2,3,6,8\}$, as follows:
\begin{align}
    \SRG{Y_{\text{grid}}^{N_i}(s)} \cap -\tau \SRG{Y_{c}^{N_i}(s)} = \emptyset .
    \label{eqn:14node_individual}
\end{align}
{Note that, as in traditional impedance-based methods \cite{Wang2024Limitations,Fan2020_problemsAdmittance,Cigre2024}, we only need to verify condition \eqref{eqn:14node_individual} once to guarantee system stability. However, we evaluate this condition from each converter PCC as this provides insights into how the grid dynamics vary from each individual connection point.}

The second modeling approach assumes knowledge about the network model. Then, it is possible to use a Kron reduction to obtain a matrix $Z_{\text{grid}}(s)\in \mathcal{RH}_\infty^{8\times8}$, and the converters in the grid are modeled as a block diagonal matrix, $\mathbf{{Y}^N_c(s)}$, where each diagonal entry is a converter admittance with its $dq$-coordinate transformation, as is shown in Fig.~\ref{fig:system}.

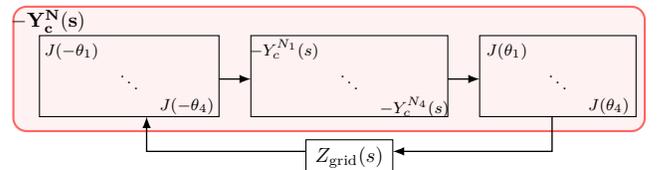
\begin{figure}[ht] 
    \centering
    \begin{tikzpicture}[scale=0.77, every node/.style={transform shape}]
\filldraw[color=red!60, fill=red!5, thick,rounded corners=5] (4.7,5.6) rectangle (15.6,3.45);
\draw (5.15,5.1) rectangle (8.25,3.7);
\draw (8.8,5.1) rectangle (12.2,3.7);
\draw (12.75,5.1) rectangle (15.45,3.7);
\node at (5.3,5.35) {$-\mathbf{{Y}^N_c(s)}$};
\node at (6.7,4.35) {\footnotesize $\begin{matrix}  J(-\theta_1)& & \\ & \ddots & \\ & &J(-\theta_4) \end{matrix}$};
\node at (10.5,4.35) {\footnotesize $\begin{matrix}  -Y_{c}^{N_1}(s)& & \\ & \ddots & \\ & &-Y_{c}^{N_4}(s) \end{matrix}$};
\node at (14.1,4.35) {\footnotesize $\begin{matrix}  J(\theta_1)& & \\ & \ddots & \\ & &J(\theta_4) \end{matrix}$};
\draw (9.75,3.3) rectangle (11.25,2.7);
\node at (10.5,3) {$Z_{\text{grid}}(s)$};
\draw[-latex, line width = .5 pt]  (14,3.7) -- (14,3.1) -- (11.25,3.1);
\draw[-latex, line width = .5 pt] (9.75,3.1) -- (7,3.1) -- (7,3.7) ;
\draw[-latex, line width = .5 pt] (8.25,4.35) -- (8.8,4.35);
\draw[-latex, line width = .5 pt] (12.2,4.35) -- (12.75,4.35);
\end{tikzpicture}
\caption{Closed-loop dynamics between $Z_{\text{grid}}(s)$ and $\mathbf{{Y}^N_c(s)}$. }
    \label{fig:system}
\end{figure}

We use Theorem~\ref{thm:GFT_nonlinear} for assessing the stability as follows:
\begin{align}
    \left(\operatorname{SRG}(Y_{\text{grid}}(s))\right) \cap -\tau \operatorname{SRG}(\mathbf{Y_{c}^{N}(s)}) = \emptyset . \label{eqn:Kron_approach}
\end{align}

Both approaches offer distinct trade-offs in practical implementation. The Kron reduction method demands an explicit model of the reduced grid, which is not always necessarily available. In contrast, the grid-impedance-estimation approach (without prior grid knowledge) measures the impedance at each PCC, inherently reflecting the aggregated dynamics of all converters. Although computationally intensive and less scalable, this method is more adaptable to actual operating conditions. \tb{Similar to most of the impedance-based approaches, the choice of one approach over the other is guided by the requirement that both subsystems must exhibit stable behavior. Nevertheless, active research efforts are underway to relax this assumption and extend the applicability \cite{krebbekx2025}.}

Fig.~\ref{fig:SRG_IEEE14node} compares $\SRG{Y_{c}^{N_i}(s)}$ and $\SRG{Y_{\text{grid}}^{N_i}(s)}$ across PCCs for the system in Fig.~\ref{fig:14_node_IEEE}. $\SRG{Y_{\text{grid}}^{N_i}(s)}$ varies significantly with location, reflecting grid dynamics: $\SRG{Y_{\text{grid}}^{N_2}(s)}$ (strong grid, near slack bus) lies farthest from the origin, while $\SRG{Y_{\text{grid}}^{N_8}(s)}$ (weak grid) is the closest, necessitating conservative GFL tuning. 

\begin{figure}[h]
\centering
\begin{subfigure}[b]{0.32\linewidth}
    \centering
    \includegraphics[width=1\linewidth]{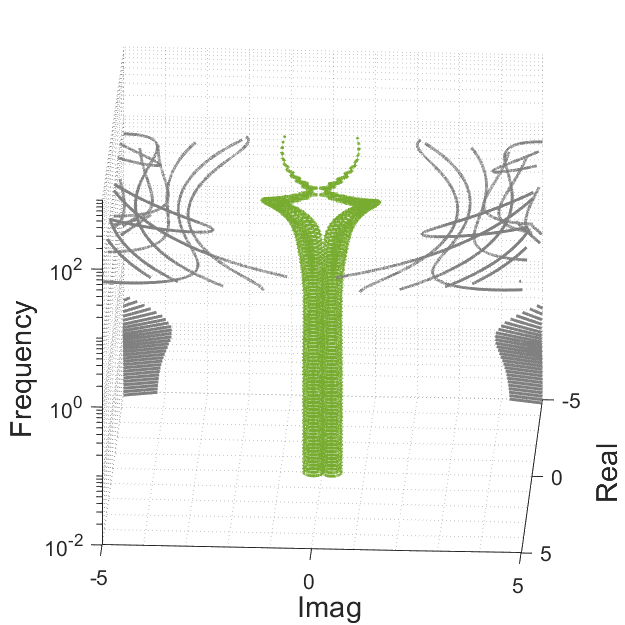}
    \caption{} 
\end{subfigure}
\begin{subfigure}[b]{0.32\linewidth}
    \centering
    \includegraphics[width=1\linewidth]{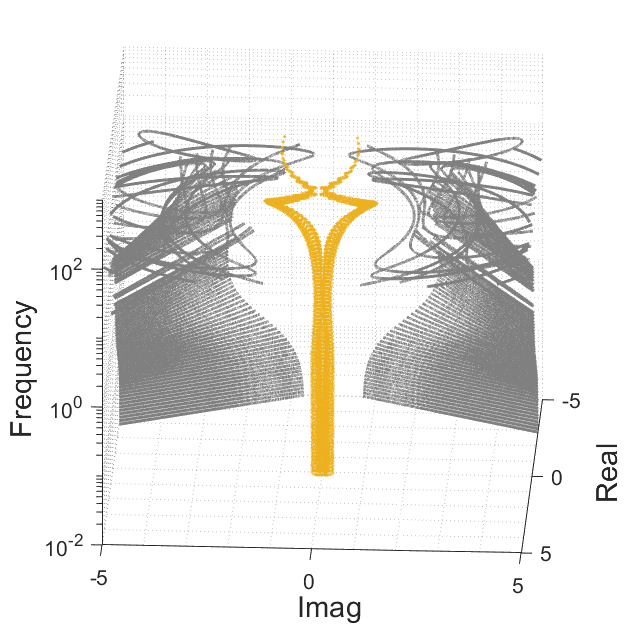}
    \caption{} 
\end{subfigure}
\begin{subfigure}[b]{0.32\linewidth}
    \centering
    \includegraphics[width=1\linewidth]{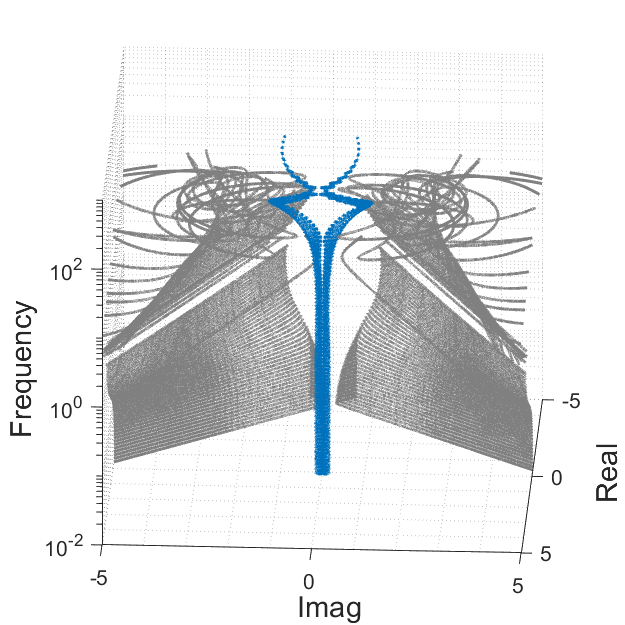}
    \caption{} 
\end{subfigure}
\begin{subfigure}[b]{0.45\linewidth}
    \centering
    \includegraphics[width=1\linewidth]{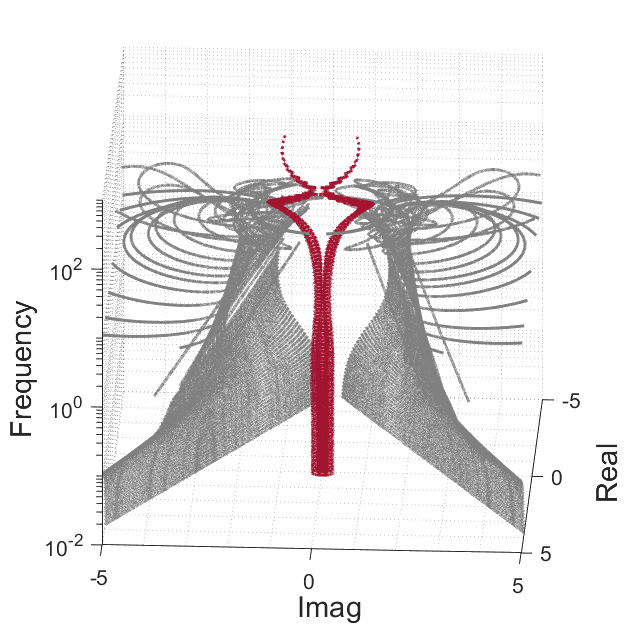}
    \caption{} 
\end{subfigure}  
\begin{subfigure}[b]{0.45\linewidth}
    \centering
    \includegraphics[width=1\linewidth]{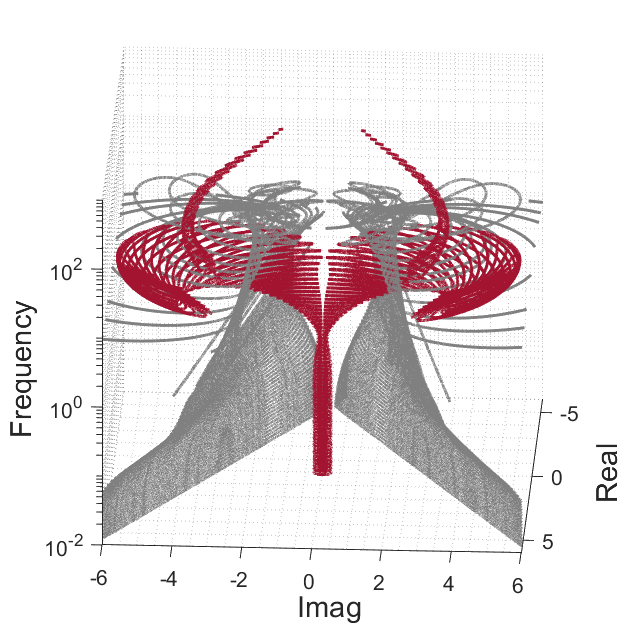}
    \caption{} 
\end{subfigure}   
\caption{(a) $\SRG{Y_c^{N_2}(s)}$ in green and $\SRG{Y_{\text{grid}}^{N_2}(s)}$ in gray.  
(b) $\SRG{Y_c^{N_3}(s)}$ in yellow and $\SRG{Y_{\text{grid}}^{N_3}(s)}$ in gray. 
(c) $\SRG{Y_c^{N_6}(s)}$ in blue and $\SRG{Y_{\text{grid}}^{N_6}(s)}$ in gray.  
(d) $\SRG{Y_c^{N_8}(s)}$ with current bandwidth  $f_{i}=600$Hz in red and $\SRG{Y_{\text{grid}}^{N_8}(s)}$ in gray. 
(e) $\SRG{Y_c^{N_8}(s)}$ with current bandwidth  $f_{i}=100$Hz in red and $\SRG{Y_{\text{grid}}^{N_8}(s)}$ in gray.}
    \label{fig:SRG_IEEE14node}
\end{figure}

\begin{figure}[ht]
    \centering
    \includegraphics[width=1\linewidth]{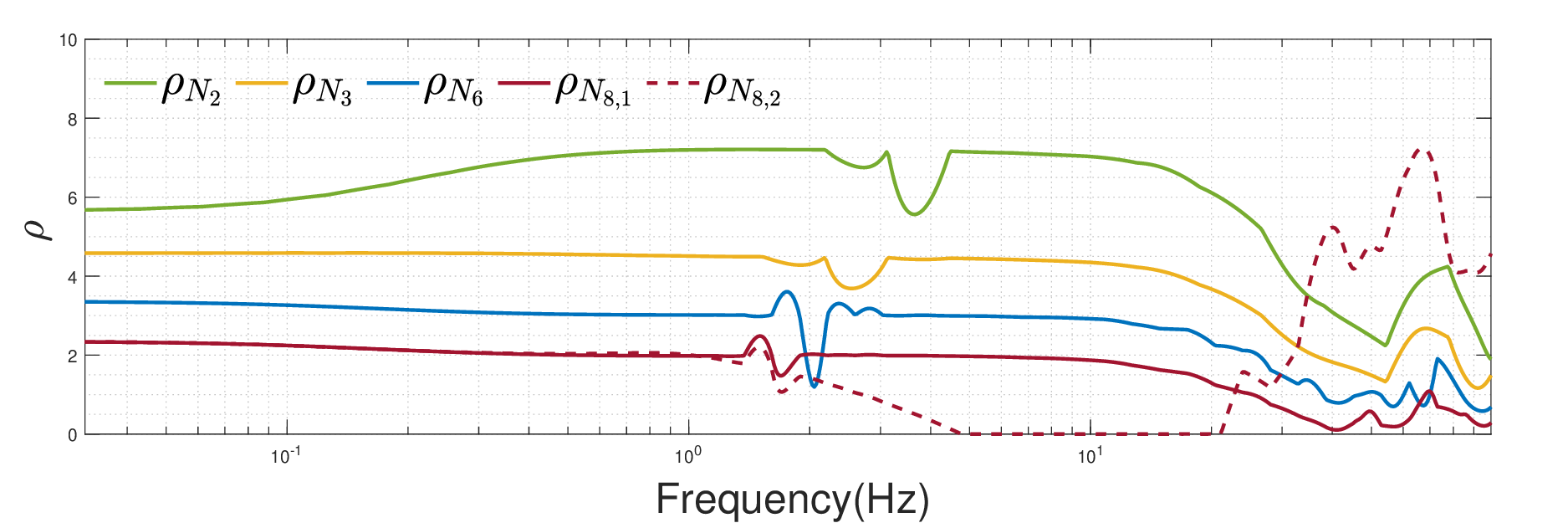}
    \caption{Stability margin for each converter on the IEEE 14 node system evaluated as \eqref{eqn:14node_individual} for both case studies.}
    \label{fig:Stabilitymargin_Model1}
\end{figure}

Figs. \ref{fig:SRG_IEEE14node}(a)-(d) show that there is no intersection between the SRG of the GFL converters and their respective equivalent network admittance, indicating system stability. This is confirmed by the stability margins $\rho_{N_{2}},\rho_{N_{3}},\rho_{N_{6}},\rho_{N_{8_1}}$ depicted in Fig.~\ref{fig:Stabilitymargin_Model1} and validated by the time domain simulations shown in Fig.~\ref{fig:Sims_IEEENode14}. 
We examine two scenarios for {converter at node~8}. Initially, the converter operates with current bandwidth $f_{i}=600$ Hz, resulting in two disjoint SRGs that satisfy Theorem~\ref{thm:GFT_nonlinear}, as depicted in Fig.~\ref{fig:SRG_IEEE14node}(d). However, when the  $f_{i}$ is decreased to $100$ Hz, the SRGs are no longer disjoint (See Fig.~\ref{fig:SRG_IEEE14node}(e)), indicating a possible unstable connection. In the $f_i=100$Hz case study, we present $\SRG{Y_{\text{grid}}^{N_{8}}(s)}$ in Fig.~\ref{fig:SRG_IEEE14node}(e) and its stability margin $\rho_{N_{8,2}}$  in Fig.~\ref{fig:Stabilitymargin_Model1} to identify the frequency bands where condition \eqref{eqn:14node_individual} is not met. Fig.~\ref{fig:Sims_IEEENode14_unstable} shows unstable $11$~Hz oscillations, aligning with Fig.~\ref{fig:Stabilitymargin_Model1} predictions.

\begin{figure}[ht]
    \centering
    \includegraphics[width=1\linewidth]{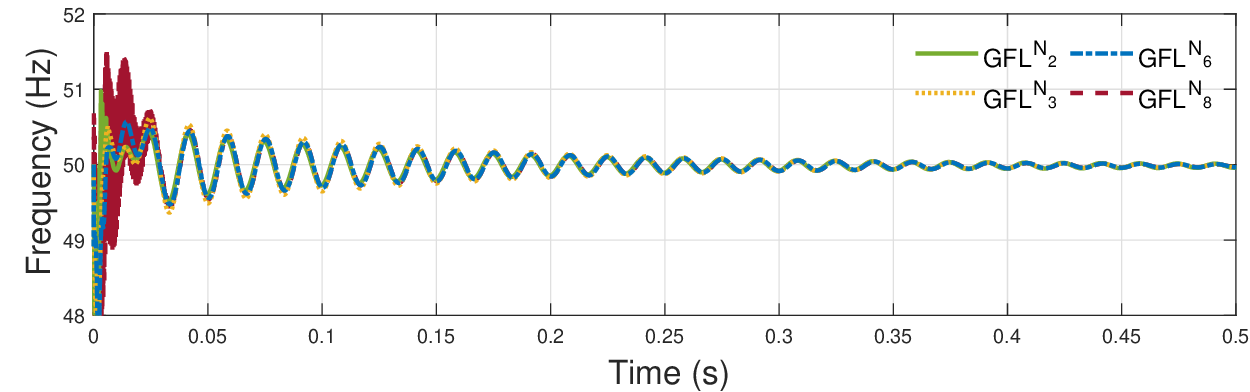}
    \caption{Frequency in the system depicted in Fig. \ref{fig:14_node_IEEE} for the GFL converters with the converter in node 8 with $f_{i}=600$Hz.}
    \label{fig:Sims_IEEENode14}
    \includegraphics[width=1\linewidth]{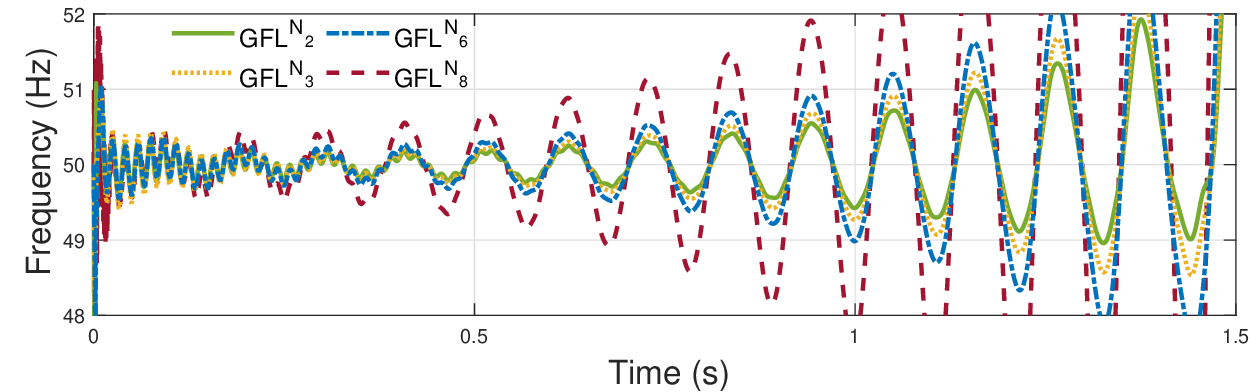}
    \caption{Frequency in the system depicted in Fig.\ref{fig:14_node_IEEE} for the GFL converters with the converter in node 8 with $f_{i}=100$Hz.}
    \label{fig:Sims_IEEENode14_unstable}
\end{figure} 

We now resort to our second modeling approach, Fig.~\ref{fig:Stability_Kron_SRG}(a) displays $\SRG{Y_{\text{grid}}(s)}$ and $\operatorname{SRG}(\mathbf{Y_{c}^{N}(s)})$ when $Y_{c}^{N_8}(s)$ is tuned with current bandwidth $f_{i} = 600$ Hz. As in the first approach, we observe no intersections between the two SRGs across the entire frequency spectrum, indicating stability. Several key insights emerge from this case study. First, $\operatorname{SRG}(\mathbf{Y_{c}^{N}(s)})$ is independent of the diagonal matrix's order, which means that the geographical placement of converters in the grid does not influence the SRG, leaving a key portion of information out of the analysis. This limitation is common in decentralized modeling and aligns with expectations for such frameworks~\cite{huang2024gain,Baron2025decentralized}. Additionally, $\operatorname{SRG}(\mathbf{Y_{c}^{N}(s)})$ is not the union of individual converter SRGs, though they are contained within it\cite{Baron2025decentralized}. When the parameter $f_{i}$ of $Y_{c}^{N_8}(s)$ is changed to $100$ Hz, {\eqref{eqn:Kron_approach} is not met at certain frequencies, meaning that the SRGs intersect each other and hence revealing a possible unstable condition. Fig.~\ref{fig:Stability_Kron_SRG}(b) depicts multiple intersections between $\operatorname{SRG}(\mathbf{Y_{c}^{N}(s)})$ and $\SRG{Y_{\text{grid}}(s)}$, along with a significant expansion of $\operatorname{SRG}(\mathbf{Y_{c}^{N}(s)})$ due to the altered $Y_{c}^{N_8}(s)$. }Finally, both feedback models assess stability, but \eqref{eqn:Kron_approach} yields conservative margins by ignoring converter locations. While computationally efficient, this simplification trades accuracy for speed, as seen when comparing the stability margins in Fig.~\ref{fig:Stabilitymargin_Model1} (individual analysis) with Fig.~\ref{fig:Stabilitymargin_Model2} (Kron reduction).

\begin{figure}[h]
    \centering
\begin{subfigure}[b]{0.24\textwidth}
    \centering 
    \includegraphics[width=0.9\linewidth]{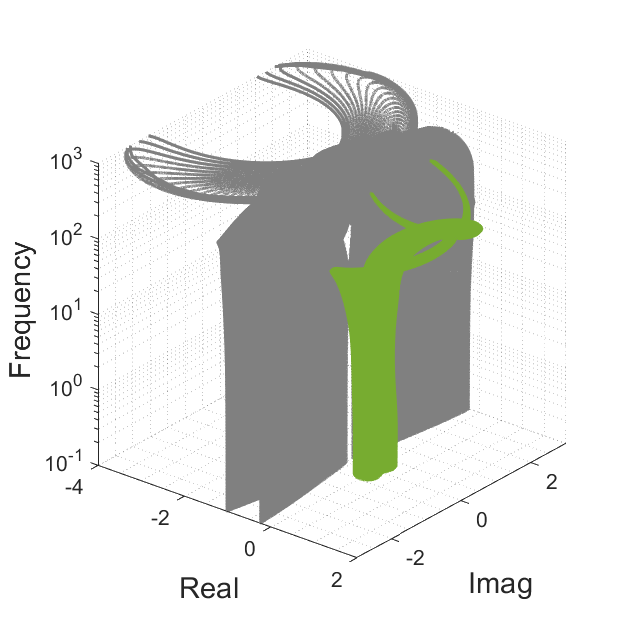}
    \caption{}
\end{subfigure}  
\begin{subfigure}[b]{0.24\textwidth}
    \centering
    \includegraphics[width=0.9\linewidth]{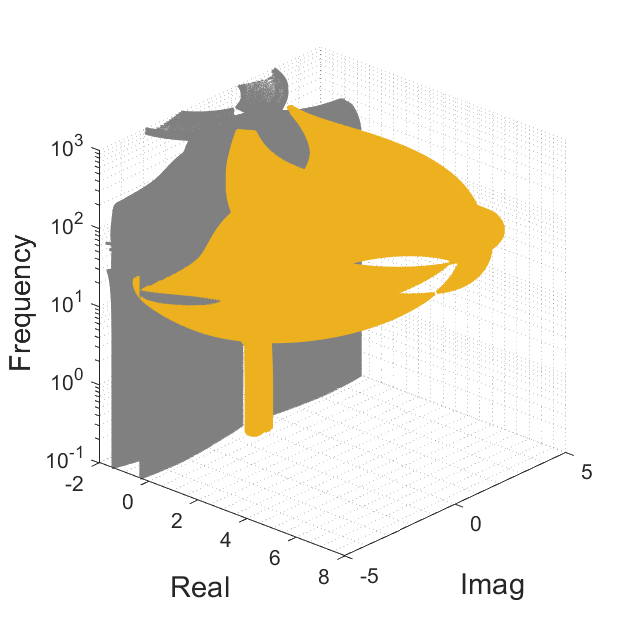}
    \caption{}
\end{subfigure}  
\begin{subfigure}[b]{0.49\textwidth}
    \centering
    \includegraphics[width=1\linewidth]{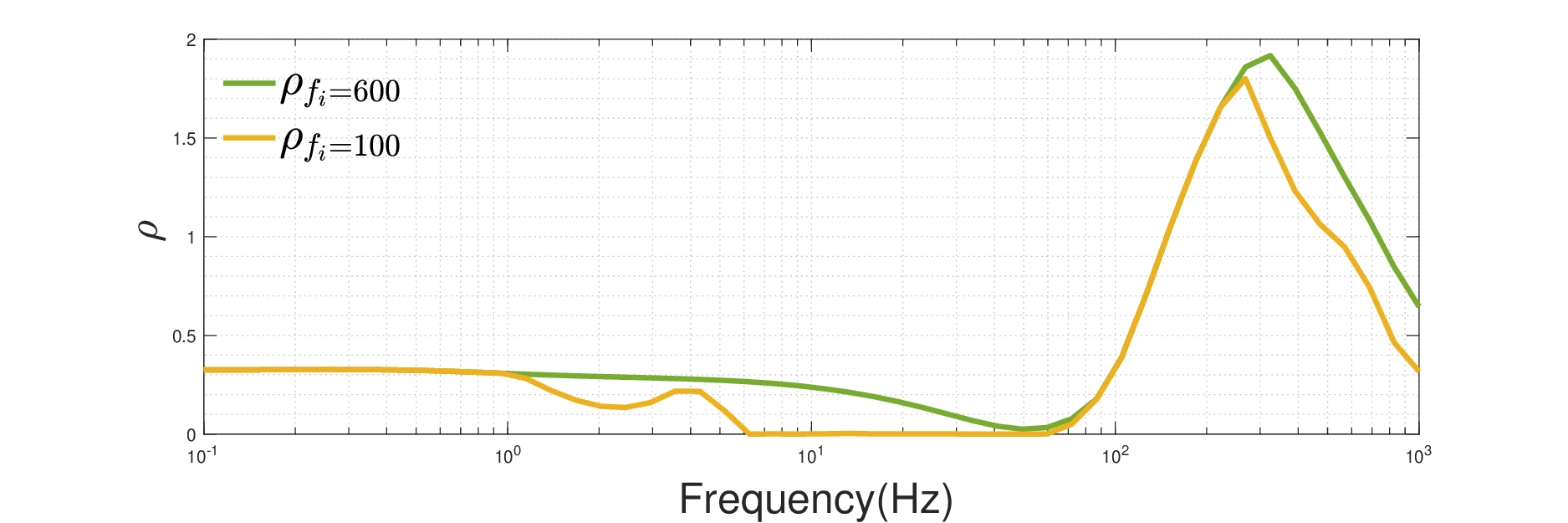}
    \caption{}
    \label{fig:Stabilitymargin_Model2}
\end{subfigure}  
    \caption{(a)  $\operatorname{SRG}(\mathbf{Y_{c}^{N}(s)})$ in green and $\SRG{Y_{\text{grid}}(s)}$ in gray. (b)  $\operatorname{SRG}(\mathbf{Y_{c}^{N}(s)})$ in green and $\SRG{Y_{\text{grid}}(s)}$ in gray for $f_{i}=600$Hz.(c) SRG-Stability margin for both cases. }
    \label{fig:Stability_Kron_SRG} 
\end{figure} 

 
\subsection{IEEE 57-bus system}
\tb{In this modified network the synchronous generators originally located at nodes~2, 6, and~12 are replaced by GFL converters, each equipped with a distinct PLL bandwidth, while a GFM converter is installed at node~8. The corresponding converter admittances are denoted as $Y_c^{N_2}(s)$, $Y_c^{N_6}(s)$, $Y_c^{N_8}(s)$, and $Y_c^{N_{12}}(s)$. We consider two scenarios that differ only in the GFL control design of $Y_c^{N_{12}}(s)$: the same controller architecture is retained, but the current controller bandwidth is modified, whereas the rest of the power system remains unchanged. Following the procedure in Subsection~\ref{subsec:14nodes}, we compute the grid admittances seen from each converter-interfaced node, denoted by $Y_{\text{grid}}^{N_2}(s)$, $Y_{\text{grid}}^{N_6}(s)$, $Y_{\text{grid}}^{N_8}(s)$, and $Y_{\text{grid}}^{N_{12}}(s)$, which characterize the effective network dynamics at each point of interconnection.}  

\begin{figure}
    \centering
    \includegraphics[width=.9\linewidth]{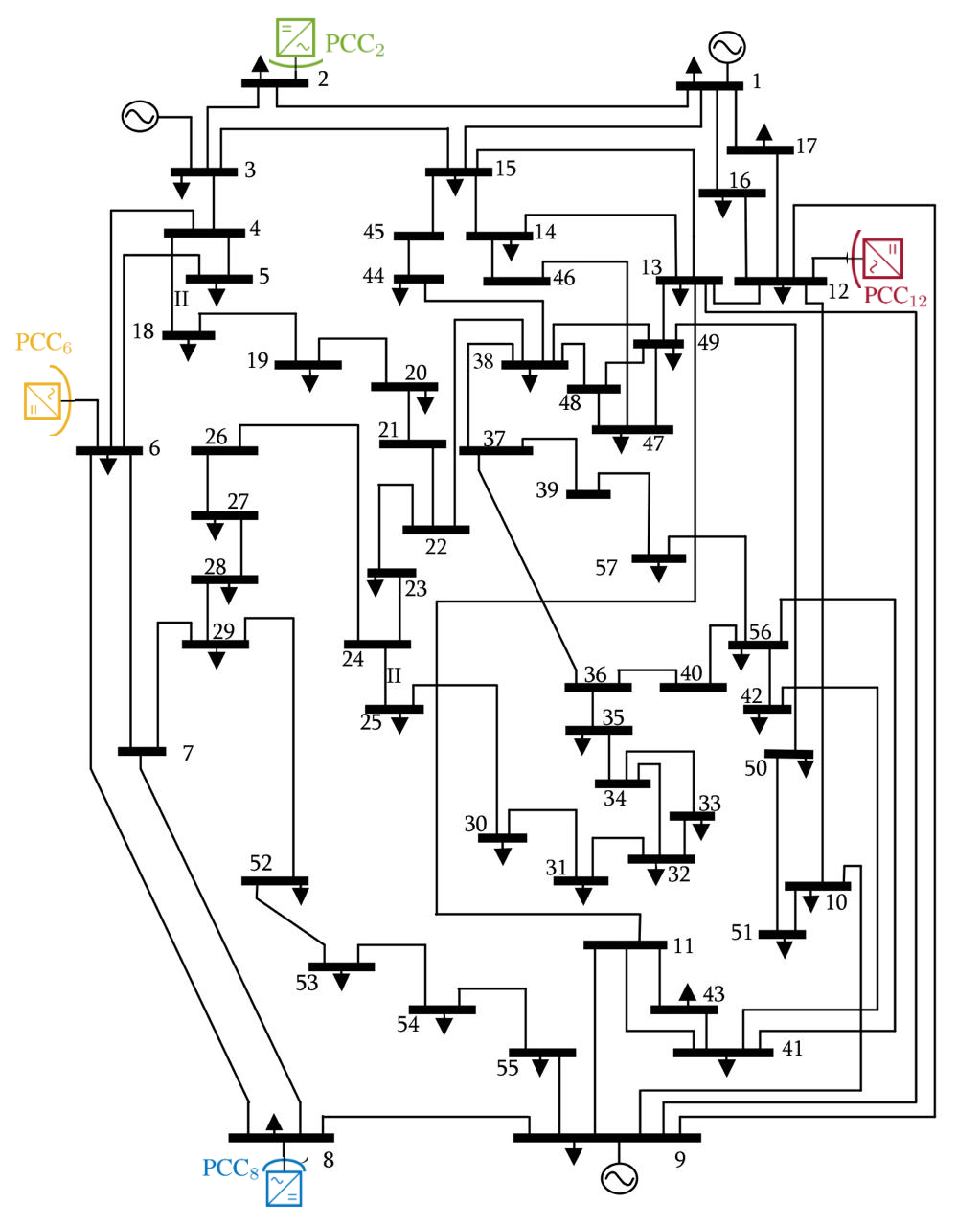}
    \caption{IEEE 57 bus test system.}
    \label{fig:57_node_IEEE}
\end{figure}
\tb{Figure~\ref{fig:SRG_IEEE57node} illustrates the application of~\eqref{eqn:14node_individual}. Each figure is rotated to facilitate the interpretability of the plot. Moreover, to enhance the perception of the stability margin, we plot the SRG stability margin in Fig.~\ref{fig:Stabilitymargin_Model_57nodes}, which reveals several intuitive trends.  Fig.~\ref{fig:Stabilitymargin_Model_57nodes} shows that the GFL converter at node~2 exhibits a notably large stability margin, primarily due to its proximity to the slack bus and its relatively low PLL bandwidth. In contrast, the GFL converters at nodes~6 and~12 display similar stability margins, despite node~12 being electrically stronger than node~6. This apparent similarity arises because the converter at node~12 employs a higher PLL bandwidth, which offsets the advantage of its stronger network position. The behavior of the GFM converter differs markedly from that of its GFL counterparts. At low frequencies ($f<5$\,Hz), it achieves substantially larger stability margins than the GFL units at nodes~6 and~12, whereas at higher frequencies its margin decreases rapidly and becomes comparable to theirs. The fact that the margin remains positive across the entire frequency range indicates that the system is stable under this configuration, which is consistent with the time-domain results in Fig.~\ref{fig:Sims_IEEENode57}, where we can observe oscillations at a high frequency in the initial part of the simulation that is attenuated in the first $5s$. However, the synchronous machine in node 3 exhibits oscillatory behavior at $ 2.4$ Hz, consistent with one of the lowest points in the SRG stability margin. In contrast, the unstable GFL configuration exhibits uniformly lower margins than the stable one, with a clear intersection near $7.3$\,Hz that coincides with the growing oscillation seen in Fig.~\ref{fig:Sims_IEEENode57_unstable}.}

\begin{figure}[bt]
\centering
\begin{subfigure}[b]{0.32\linewidth}
    \centering
    \includegraphics[width=1\linewidth]{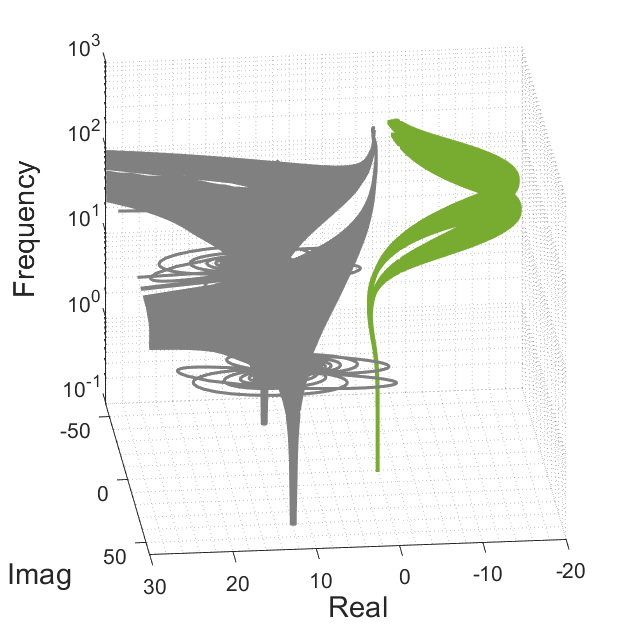}
    \caption{} 
\end{subfigure}
\begin{subfigure}[b]{0.32\linewidth}
    \centering
    \includegraphics[width=1\linewidth]{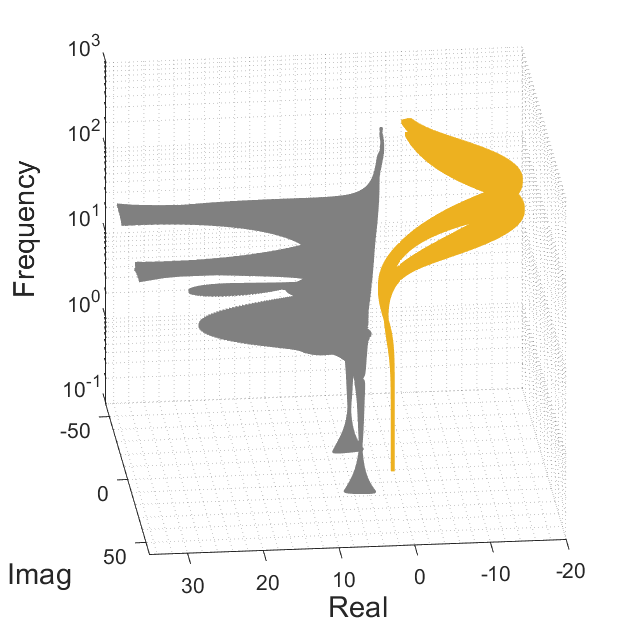}
    \caption{} 
\end{subfigure}
\begin{subfigure}[b]{0.32\linewidth}
    \centering
    \includegraphics[width=1\linewidth]{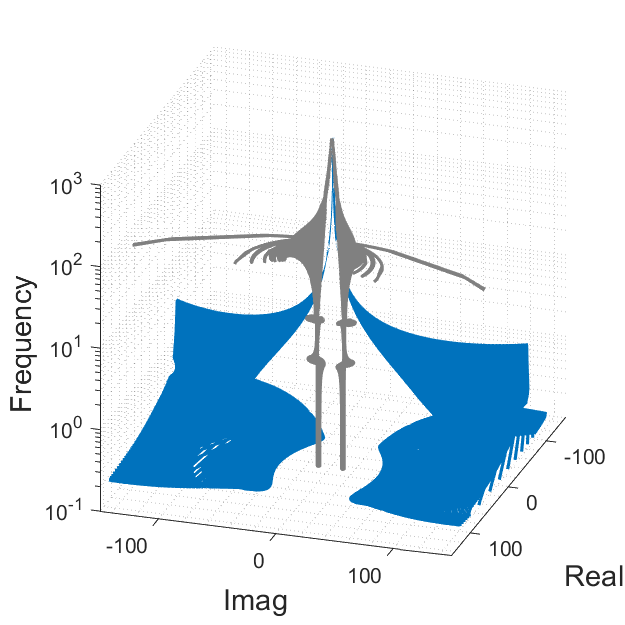}
    \caption{} 
\end{subfigure}
\begin{subfigure}[b]{0.49\linewidth}
    \centering
    \includegraphics[width=1\linewidth]{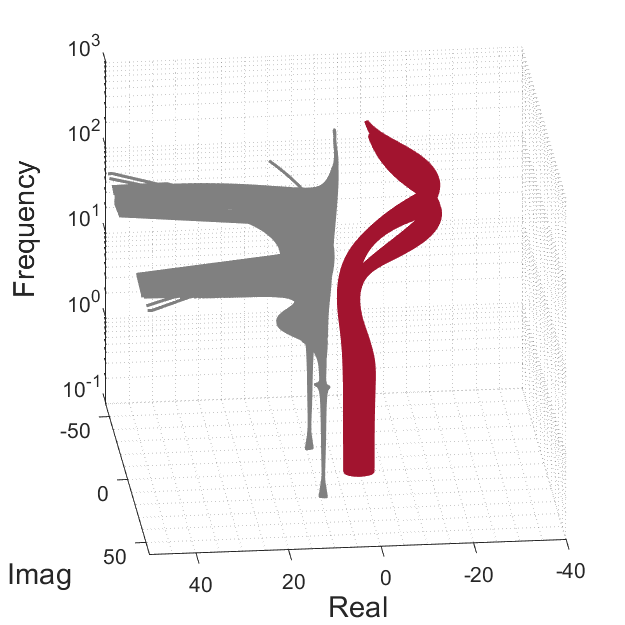}
    \caption{} 
\end{subfigure}
\begin{subfigure}[b]{0.49\linewidth}
    \centering
    \includegraphics[width=1\linewidth]{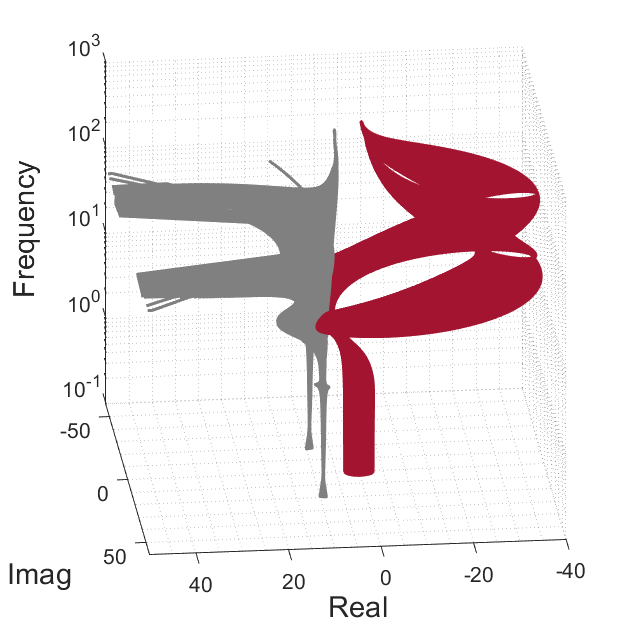}
    \caption{} 
\end{subfigure}
\caption{\tb{(a) $\SRG{Y_c^{N_2}(s)}$ in green and $\SRG{Y_{\text{grid}}^{N_2}(s)}$ in gray.  
(b) $\SRG{Y_c^{N_6}(s)}$ in yellow and $\SRG{Y_{\text{grid}}^{N_6}(s)}$ in gray. 
(c) $\SRG{Y_c^{N_8}(s)}$ in blue and $\SRG{Y_{\text{grid}}^{N_8}(s)}$ in gray.  
(d) $\SRG{Y_c^{N_{12}}(s)}$ with current bandwidth  $f_{i}=300$Hz in red and $\SRG{Y_{\text{grid}}^{N_{12}}(s)}$ in gray. 
(e) $\SRG{Y_c^{N_{12}}(s)}$ with current bandwidth  $f_{i}=100$Hz in red and $\SRG{Y_{\text{grid}}^{N_{12}}(s)}$ in gray.}}
    \label{fig:SRG_IEEE57node}
\end{figure} 

\begin{figure}[ht]
    \centering
    \includegraphics[width=1\linewidth]{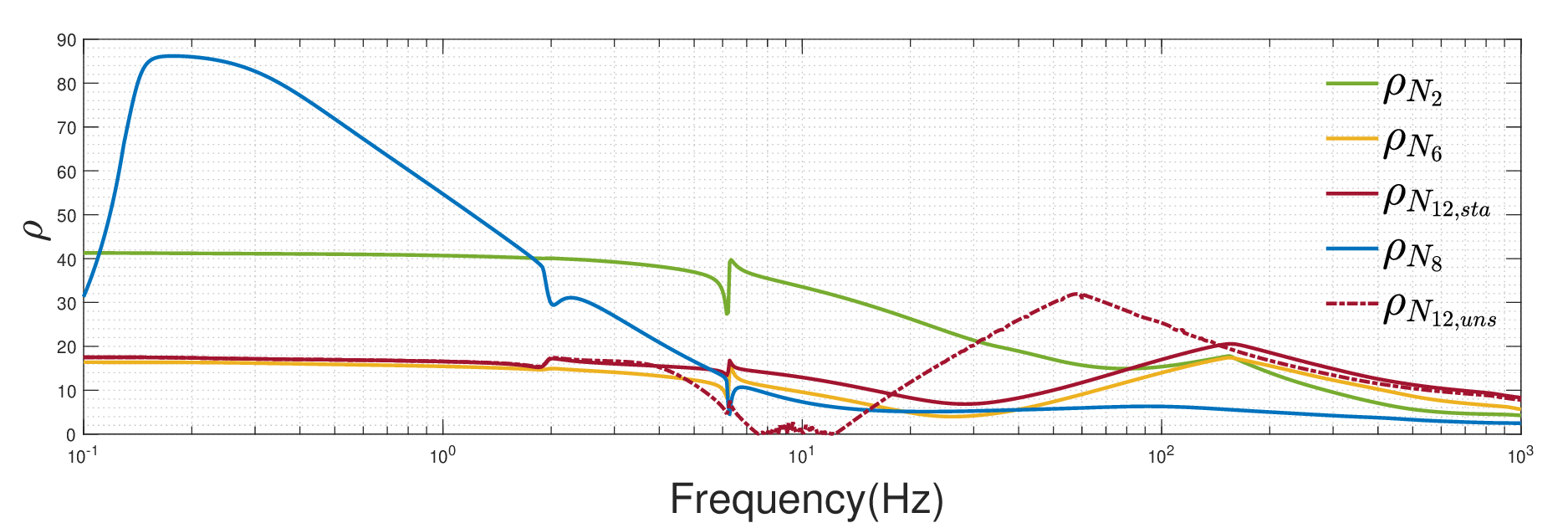}
    \caption{Stability margin for each converter on the IEEE 57 node system evaluated as \eqref{eqn:14node_individual} for both case studies.}
    \label{fig:Stabilitymargin_Model_57nodes}
\end{figure} 
\begin{figure}[ht]
    \centering
    \includegraphics[width=1\linewidth]{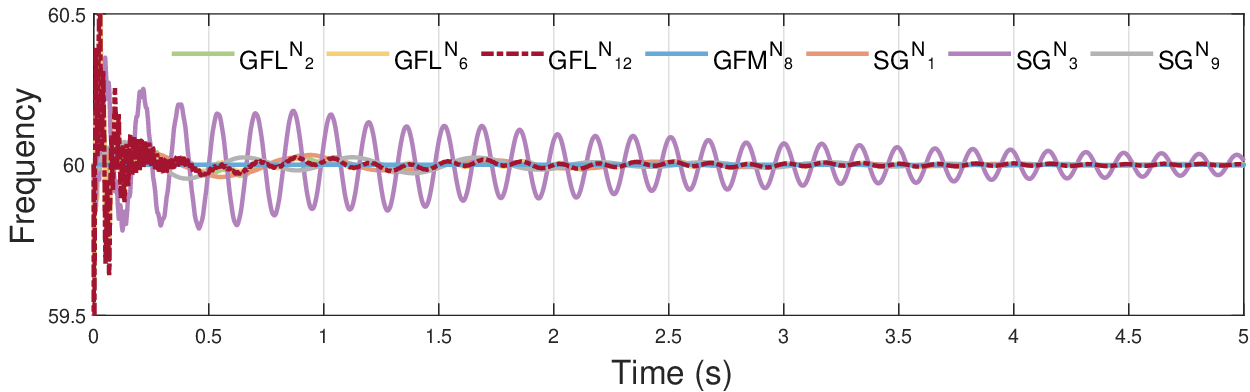}
    \caption{Frequency in the system depicted in Fig. \ref{fig:57_node_IEEE} for the GFL converters with the converter in node 12 corresponding to a stable GFL configuration.}
    \label{fig:Sims_IEEENode57}
    \includegraphics[width=1\linewidth]{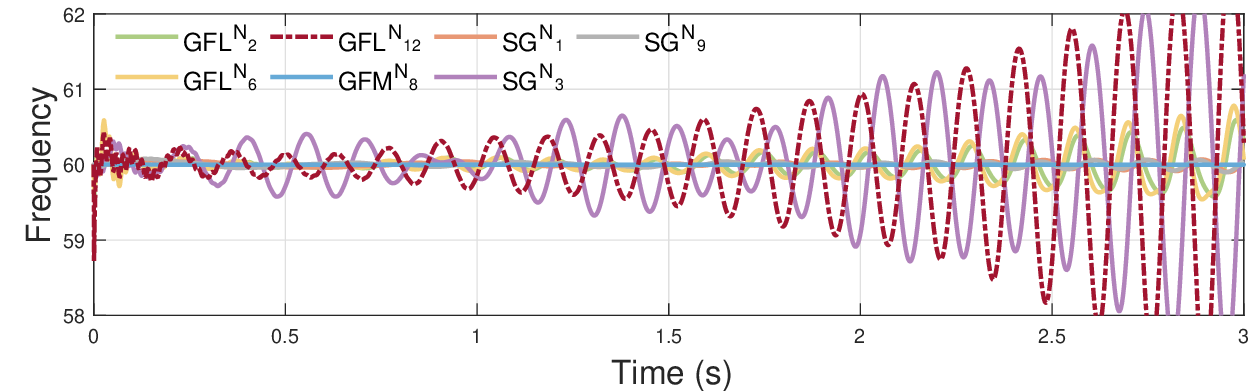}
    \caption{Frequency in the system depicted in Fig.\ref{fig:57_node_IEEE} for the GFL converters with the converter in node 12 corresponding to a unstable GFL configuration.}
    \label{fig:Sims_IEEENode57_unstable}
\end{figure}

\begin{remark}[Decentralized stability conditions]
\label{rem:decentralized_setting}
\tb{The underlying SRG formulation is fully amenable to decentralization as shown in \cite{Baron2025decentralized}. Therefore, even though this manuscript does not exploit that structure, the proposed method remains compatible with decentralized operation and can be extended in that direction with minimal modifications. In such a case, the grid is represented by its Kron reduction, and its SRG is over-approximated at each frequency by a disk, denoted $\operatorname{SRG}(\widehat{Y_{\text{grid}}}(s))$. If for each converter $i$ and for each frequency $\omega\in[0,\infty)$ 
\begin{align*}
    \operatorname{SRG}(\widehat{Y_{\text{grid}}}(s)) \;\cap\; -\tau\,\operatorname{SRG}\!\left(Y_{c}^{N_i}(s)\right) = \emptyset,\; \forall \tau\in[0,1],
\end{align*}
then, the system is stable. Further details on the decentralized formulation are provided in~\cite{Baron2025decentralized}.}
\end{remark} 

\section{Stability Analysis Methods: A Comparison } \label{sec:comparison}
\tb{The proposed framework is rooted in the impedance-based, frequency-domain viewpoint, where stability is assessed through input–output relations rather than internal state dynamics. Hence, a comparison with state-space approaches is not pursued, as these rely on full system white-box models which are rarely available. Our method instead extends classical frequency-domain criteria (such as Nyquist and passivity tests) toward a geometric characterization, making comparisons meaningful only within this family of impedance-based formulations.} We now analyze state-of-the-art \tb{impedance-based} stability conditions and compare them to Theorem \ref{thm:GFTLTI}. We start by briefly reviewing the most used stability tools in industrial and academic environments. First, the GNC 
considers $\Tilde{Y}_{c}(s), Z_{\text{grid}}(s) \in \mathcal{RH}_\infty^{2\times 2}$ connected as shown in Fig. \ref{fig:decentralizedfb}, with a well-posed system interconnection. The closed-loop system achieves exponential stability if and only if:
\begin{align}
    \operatorname{det}(I+\Tilde{Y}_{c}(s)Z_{\text{grid}}(s)) \neq 0, \quad \forall s=\textup{j}\omega, \; \omega \in \mathbb{R}, \label{eqn:NyquistGC}
\end{align}
and the winding number of $\det(I+\Tilde{Y}_{c}(s)Z_{\text{grid}}(s))$ about the origin vanishes (See details in \cite{skogestad2005}).

As a second condition, the mixed gain-phase theorem combines two fundamental stability theorems. On one hand, the gain condition uses the maximum singular value ($\sigma_{\max}(\cdot)$ representing the maximum gain at each frequency) to require: 
\begin{align}
    \sigma_{\max}(\Tilde{Y}_{c}(s))\sigma_{\max}(Z_{\text{grid}}(s)) < 1. \label{eqn:smallgain}
\end{align}
     
The phase condition uses the maximum $\alpha_{\max}(\cdot)$ and minimum phase angle $\alpha_{\min}(\cdot)$, denoting the largest phase shifts at every frequency to enforce: 
\begin{subequations}\label{eqn:smallphase}
   \begin{align} 
 \alpha_{\max}(\Tilde{Y}_{c}(s)) + \alpha_{\max}(Z_{\text{grid}}(s)) &< \pi,\\
  \alpha_{\min}(\Tilde{Y}_{c}(s)) + \alpha_{\min}(Z_{\text{grid}}(s)) &> -\pi.
\end{align} 
\end{subequations}

The small phase criterion requires the semi-sectorial property \cite{Wang2024}, where an operator $A$ is said to be semi-sectorial if it has a numerical range with supporting lines with an angle less or equal than $\pi$.

Finally, passivity-based stability analysis provides another perspective. Defining the input feedforward passivity (IFP) and output feedback passivity (OFP) indices as:
\begin{align}
    \operatorname{IFP}(Y_c(\textup{j}\omega)) &= \nicefrac{1}{2} \min \lambda[Y_c(\textup{j}\omega)+Y_c^H(\textup{j}\omega)], \label{eqn:IFP_def} \\
    \operatorname{OFP}(Y_c(\textup{j}\omega)) &= \nicefrac{1}{2}\min \lambda[Y_c^{-1}(\textup{j}\omega)+Y_c^{-H}(\textup{j}\omega)], \label{eqn:OFP_def}
\end{align}
where $\min\lambda$ denotes the minimum eigenvalue,  the closed-loop system shown in Fig. \ref{fig:decentralizedfb} is exponentially stable if $\forall s=\textup{j}\omega$ with $\omega \in [0,\infty)$\cite{Wang2024Limitations}, either
\begin{align}
    \operatorname{IFP}(\Tilde{Y}_{c}(s)) + \operatorname{OFP}(Z_{\text{grid}}(s)) &> 0, \quad \text{or} \label{eqn:IFP} \\
    \operatorname{OFP}(\Tilde{Y}_{c}(s)) + \operatorname{IFP}(Z_{\text{grid}}(s)) &> 0. \label{eqn:OFP}
\end{align} 
\begin{remark}[Difference of Theorem \ref{thm:GFT_nonlinear} and reviewed methods]
Theorem \ref{thm:GFT_nonlinear} offers a distinct advantage over the reviewed methods by allowing direct inclusion of CPL in the stability assessment. The approximation of a CPL as a gain-bounded operator in
Lemma \ref{lemma:approx_SRG} can be used for small-gain analysis\label{remark:CPL}.
\end{remark}

\subsection{Case Study 3: Comparison}
To compare the stability methods, we study a  PLL-synchronized GFL converter as an example. The input admittance matrix expressions, denoted as $Y_c(s)$, are given in \cite{chen2022impedance,Wang2024Limitations}, with used parameters that can be found in Appendix \ref{appendix:comparison}. The system is depicted in Fig. \ref{fig:setup_single}, and we consider no CPL in this context in line with Remark \ref{remark:CPL}. We assume the DC side is ideal and the AC grid is modeled as $$Y_{\text{grid}}(s)=\begin{bmatrix}
        R+sL     & -\omega_0L \\
        \omega_0L & R+sL
    \end{bmatrix}^{-1}.$$

The system's stability is confirmed by both the Nyquist and SRG methods. While the Nyquist plot in Fig.~\ref{fig:Comparison_Smallsignal}(a) offers a more straightforward interpretation than the SRG-based plot in Fig.~\ref{fig:Comparison_Smallsignal}(b), Nyquist plots can become unwieldy when analyzing power grids, as shown in \cite{Cigre2024,Fan2020_problemsAdmittance}. To provide more comprehensive insights into system behavior, we present half of the SRG in Fig.~\ref{fig:Comparison_Smallsignal}(b). Given the SRG's symmetrical nature, this partial view offers better clarity for the reader. {The $\operatorname{SRG}(Y_c(s))$ establishes a criterion for determining grid-converter compatibility based on the decoupling principle in Theorem \ref{thm:GFTLTI}. This framework is particularly valuable for modern grids with rapidly changing admittance characteristics, such as those with high renewable energy integration. Our approach assists operators and manufacturers in evaluating converter stability under such variable conditions.}

\begin{figure}[ht]
\begin{subfigure}[b]{0.24\textwidth}
    \centering
    \includegraphics[width=1\linewidth]{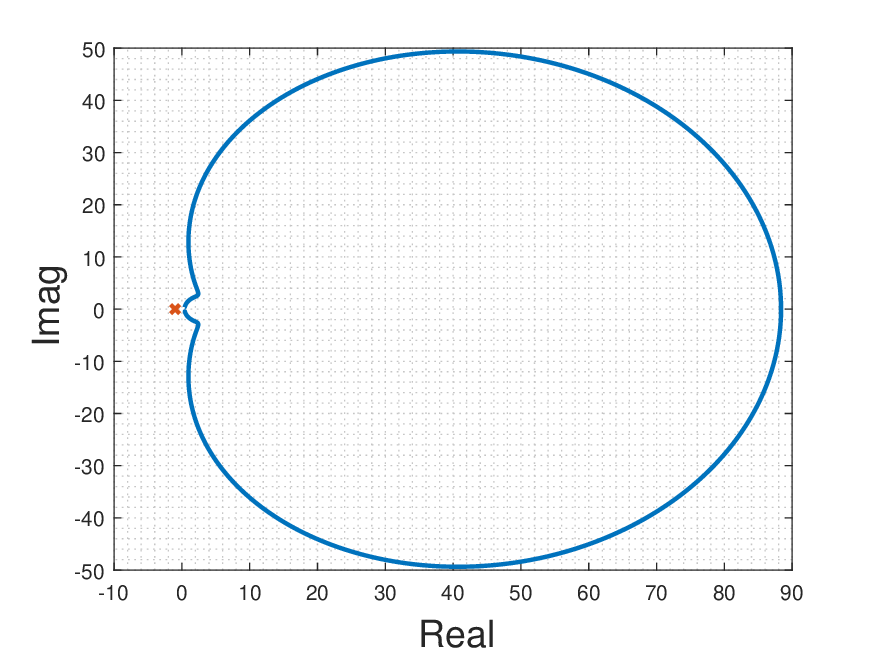}
    \caption{}
\end{subfigure}
\begin{subfigure}[b]{0.24\textwidth}
    \centering
    \includegraphics[width=1\linewidth]{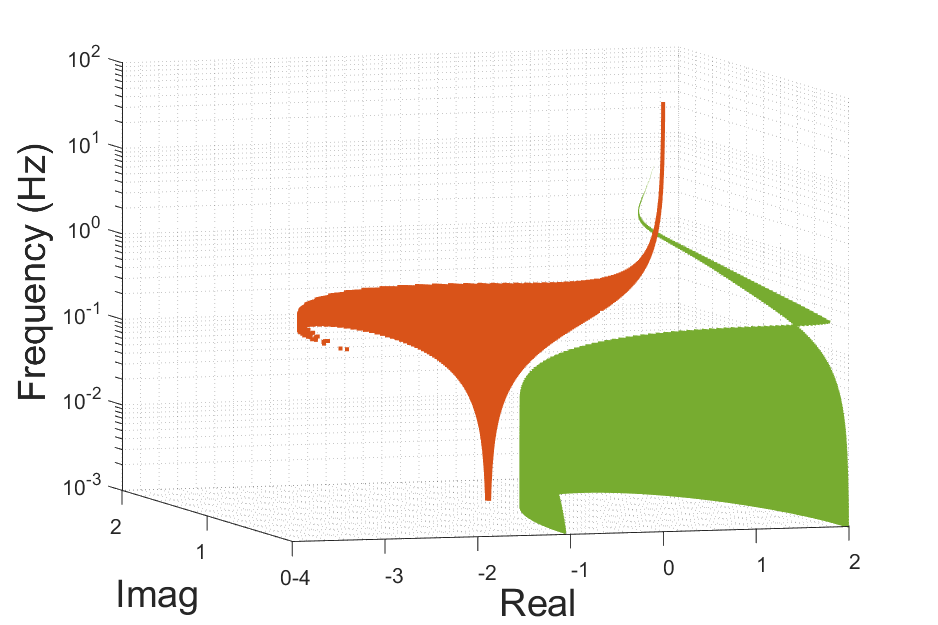}
    \caption{} 
\end{subfigure}
\begin{subfigure}[b]{0.24\textwidth}
    \centering
    \includegraphics[width=1\linewidth]{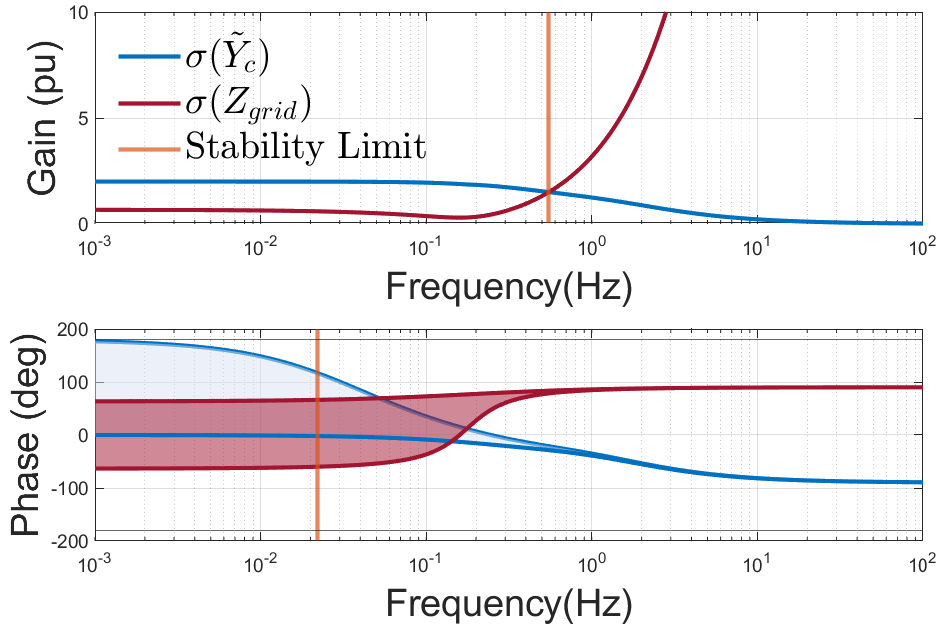}
    \caption{}
\end{subfigure}
\begin{subfigure}[b]{0.24\textwidth}
    \centering
    \includegraphics[width=1\linewidth]{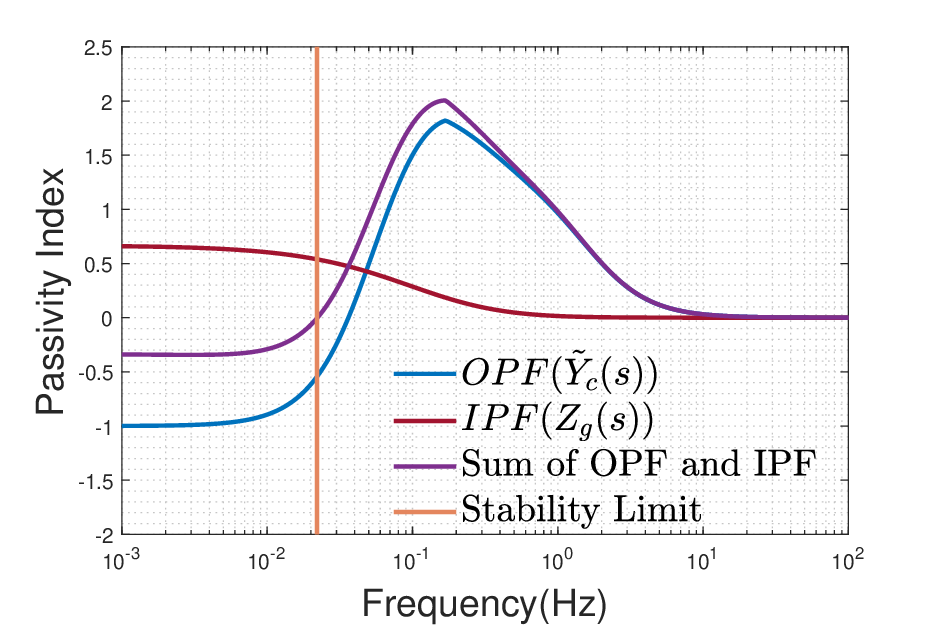}
    \caption{} 
\end{subfigure}
\begin{subfigure}[b]{0.5\textwidth}
    \centering
    \includegraphics[width=1\linewidth]{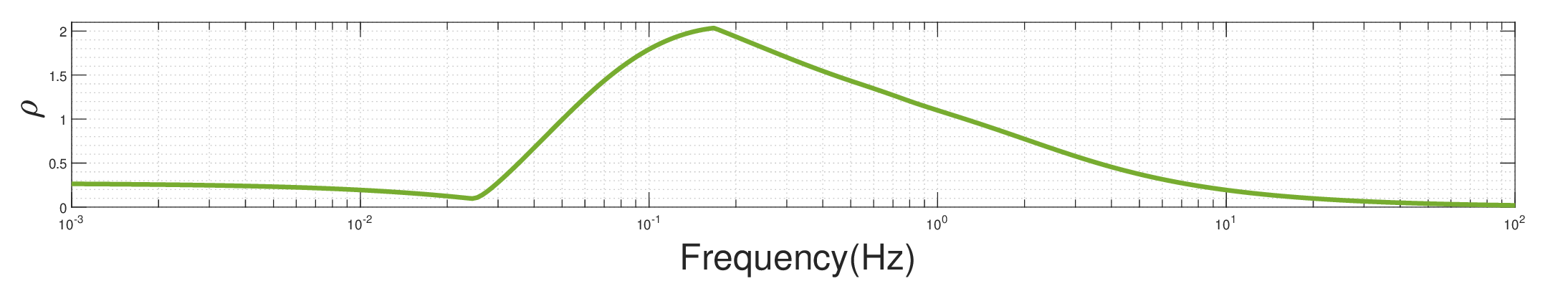}
    \caption{} 
\end{subfigure}  
    \caption{(a) Nyquist plot of $(I+\Tilde{Y}_{c}(s)Z_{\text{grid}}(s))$(b) $\SRG{Y_c(s)}$ in green with -$\SRG{Y_{\text{grid}}(s)}$ in orange. (c) Above $\sigma_{max}(\Tilde{Y}_{c}(s))$ in blue with $\sigma_{max}(Z_{\text{grid}}(s))$ in red. Below $\alpha_{\max}(\Tilde{Y}_{c}(s))$ in blue and $\alpha_{\max}(Z_{\text{grid}}(s))$ in red. (d) IFP($\Tilde{Y}_{c}(s)$) in blue and OFP($Z_{\text{grid}}(s)$) in red, in magenta  \eqref{eqn:IFP} (e) SRG-Stability margin.}
    \label{fig:Comparison_Smallsignal}
\end{figure}

Fig.~\ref{fig:Comparison_Smallsignal}(c) shows the gain and phase plots. In this case, the small gain theorem certifies the stability of the feedback system from $f\in (0.548,\infty]$ Hz, and the small phase theorem can certify stability in the range of $f\in (0.022,\infty]$ Hz. Therefore, the mixed small gain and phase theorem fails to assess the stability of the GFL converter between $f\in (0,0.022]$ Hz. In the case of the analysis with the passivity theorem, we reach the results shown in Fig.~\ref{fig:Comparison_Smallsignal}(d). It is possible to observe that the condition in \eqref{eqn:IFP} fails between $f\in (0,0.024]$ Hz, making it unsuitable to characterize the converter connection's stability.  {Fig.\ref{fig:Comparison_Smallsignal}(e) shows a low (but positive) SRG-stability margin for $f\in(0.01,0.022]$ Hz, which is the frequency band where the mixed small gain and phase and passivity theorems fail to assess stability.}

\subsection{Complexity and Conservatism of the Methods}
 
\tb{We compare the computational cost and conservatism of the previous methods for the closed-loop connection in Fig.~\ref{fig:decentralizedfb}, where each subsystem has an $n\times n$ admittance evaluated over $N_\omega$ frequency points.}

\tb{In the SRG method, each subsystem is processed independently. The SRG boundary is obtained from the numerical range of $Y(\mathrm{j}\omega)$ via the dominant eigenvalues of the Hermitian matrices $H_\phi=\tfrac{1}{2}(e^{-\mathrm{j}\phi}Y+e^{\mathrm{j}\phi}Y^H)$ for $N_\phi$ angular samples. Each eigenvalue computation costs $\mathcal{O}(n^3)$, giving $\mathcal{O}(N_\omega N_\phi n^3)$ overall. The frequency-wise decoupling enables full parallelization and requires no loop formation.}

\tb{The generalized Nyquist criterion constructs the loop $L(s)=\widetilde{Y}_{\mathrm{c}}(s)Z_{\mathrm{grid}}(s)$ and evaluates $\det(I+L(\mathrm{j}\omega))$ or its eigenvalues, with each matrix multiplication or eigenvalue computation costing $\mathcal{O}(n^3)$. This yields $\mathcal{O}(N_\omega n^3)$ complexity but couples all subsystems and requires detailed models. SRG-based analysis, in contrast, can use over-approximated or identified SRGs admittances, making it amenable when there is no full modeling data available, while the GNC remains exact and non-conservative for LTI systems.} 

\tb{Passivity-index, small-gain, and small-phase tests have similar cost $\mathcal{O}(N_\omega n^3)$ and share the SRG decoupled structure. However, their applicability is limited, since converter admittances typically violate passivity near the synchronous frequency\cite{Wang2023_Passivity}, small gain is typically satisfied only for low frequency ranges, and small-phase conditions become ill-defined when the admittance is non-sectorial\cite{huang2024gain}.}

\tb{Overall, all methods scale cubically with subsystem dimension. The SRG approach is distinguished by its parallelizable structure, ability to work with model over-approximations, and provision of a frequency-wise stability margin rather than a binary test, offering clearer geometric insight. The GNC is exact but it is centralized and requires detailed model knowledge, while passivity and gain/phase-based methods are efficient yet often conservative in low-frequency regions where converter admittances deviate from ideal passive or sectorial behavior.}

\section{Conclusions and Future Work} \label{sec:conclusions}

This paper presents a novel stability certification method for grid-connected converters based on the SRG framework. The proposed approach provides several key advantages: it enables decoupled analysis of grid and converter dynamics, offering a unified stability framework applicable to both GFL and GFM converters. Furthermore,  it incorporates CPLs directly into the stability analysis without requiring linearization and using easily computable bounds. Future research will focus on developing SRG-based methods using black-box modeling.
\bibliographystyle{IEEEtran}
\bibliography{bibtex/TIE.bib}
\appendix
\subsection{Appendix: Proof Lemma \ref{lemma:approx_SRG}}\label{appendix:approximation}
Consider two arbitrary signals $ v_1, v_2$ with amplitudes larger than $ v_{\min} $. Substituting them into \eqref{eqn:SRG_operator_nonlinear} yields: 
\begin{small}
 \begin{align*}
 \operatorname{SRG}(y_{cp})=\hspace{7cm} \\\left\{\dfrac{\left\| \frac{Mv_2}{\|v_2\|_2^2} - \frac{Mv_1}{\|v_1\|_2^2} \right\|_2}{\| {v_2-v_1}\|_2} \exp \left[\pm \textup{j} \angle\left(v_2-v_1,\frac{Mv_2}{\|v_2\|_2^2} - \frac{Mv_1}{\|v_1\|_2^2}\right) \right]\right\},
\end{align*}   
\end{small}
We focus on bounding the gain. We observe that:
\begin{align}
    \dfrac{\left\| \frac{Mv_2}{\|v_2\|_2^2} - \frac{Mv_1}{\|v_1\|_2^2} \right\|}{\| {v_2-v_1}\|_2}\leq {\sigma_{\max}(M)} \dfrac{\left\| \frac{v_2}{\|v_2\|_2^2} - \frac{v_1}{\|v_1\|_2^2} \right\|_2}{\| {v_2-v_1}\|_2}.\label{eqn:step1_approximation}
\end{align}

Define $f(v) =  \frac{v}{\|v\|_2^2}$, for $\norm{v}>0$. The Jacobian is given by 
\begin{align*}
    Df(v) = \frac{I}{\|v\|_2^2} - \frac{2vv^\top }{\|v\|_2^4},
\end{align*}
and the induced $2$-norm of $Df(v)$ is
\begin{align*}
   \sup_{\norm{u}=1} \|Df(v)u\|_2 = \left\|\frac{u}{\|v\|_2^2}-\frac{2(v^\top  u)v}{\|v\|_2^4}\right\|_2 = \frac{1}{\|v\|^2}.
\end{align*}

By the mean value inequality, and $ \norm{v}_2 \geq v_{\min} $,  we reach
\begin{align}
\|f(v_2)-f(v_1)\|_2 &\leq \sup_{v\in[v_1,v_2]}{\|Df(v)\|_2}\|v_2-v_1\|_2 \nonumber\\
\dfrac{\|f(v_2)-f(v_1)\|_2}{\|v_2-v_1\|_2}&\leq \frac{1}{{v^2_{\min}}}.    \label{eqn:step2_approximation}
\end{align}    
Thus, replacing \eqref{eqn:step2_approximation} into \eqref{eqn:step1_approximation}, we reach to
\begin{align*}
\dfrac{\left\| \frac{Mv_2}{\|v_2\|_2^2} - \frac{Mv_1}{\|v_1\|_2^2} \right\|_2}{\| {v_2-v_1}\|_2} \leq \sigma_{\max}(M)\dfrac{\|f(v_2)-f(v_1)\|_2}{\| {v_2-v_1}\|_2} \leq \frac{\sigma_{\max}(M)}{{v^2_{\min}}},
\end{align*}
which bounds the gain. Since we do not constrain the phase angles of $\SRG{y_{cp}}$, we consider the worst-case scenario where the SRG covers the full angular range $[-\pi, \pi]$, leading us to \eqref{eqn:bound_CPL}.

\subsection{Appendix: Proof Lemma \ref{lemma:CPL_eps_volterra}}
\label{appendix:lemmaCPL}

\tb{We can write $\frac{1}{\|v\|_2^2}=\frac{1}{V_0^2}\tfrac{1}{1+\xi}$ with $\xi=2\frac{v_0^\top v_\delta}{V_0^2}+\frac{ v_\delta^\top v_\delta}{V_0^2}$. For $\rho<0.1$, $\|\xi\|_\infty<1$, then the Neumann series converges uniformly as $\frac{1}{1+\xi}=\sum_{k=0}^\infty(-1)^k\xi^k.$ 
Thus, replacing into \eqref{eqn:Zgrid_CPL}, we reach to 
\begin{align*}
    i=i_{\mathrm{lin}}+i_{\mathrm{har}}=\frac{1}{V_0^2}\sum_{k=0}^\infty(-1)^k\xi^k\,M(v_0+ v_\delta).
\end{align*} 
The term with $k=0$ is linear in $v_0$ and $v_\delta$, which allow us to build the frequency preserving response as $i_{\mathrm{lin}}=\frac{1}{V_0^2}M(v_0+ v_\delta)$. Therefore, the nonlinear remainder is }
\tb{\begin{align*}
    i_{\text{har}}=\frac{1}{V_0^2}\sum_{k=1}^\infty(-1)^k\xi^k\,M(v_0+ v_\delta).
\end{align*}}

\tb{Then, $\|i_{\text{har}}\|_2\le\frac{\sigma_{\text{max}}(M)}{V_0^2}\|\sum_{k=1}^\infty\xi^k\,v\|_2$. Using the triangular inequality, $\|i_{\text{har}}\|_2\le\frac{\sigma_{\text{max}}(M)}{V_0^2}\sum_{k=1}^\infty\|\xi^k\,v\|_2$. Recalling that  $\|xy\|_2\leq\|x\|_\infty\|y\|_2$,  we can bound $\|\xi^k\,v\|_2\leq\|\xi\|^{k-1}_\infty \|\xi\|_2 \|v\|_\infty$. Therefore, given that $\|\xi\|_2\leq\tfrac{2+\rho}{V_0}\|v_\delta\|_2$, $\|\xi\|_\infty \leq 2\rho+\rho^2$ and $\|v\|_\infty\leq{V_0(1+\rho)}$
\begin{align*}
    \|i_{\text{har}}\|_2\le\frac{\sigma_{\text{max}}(M)}{V_0^2}\sum_{k=1}^\infty\|\xi\|^{k-1}_\infty (2+\rho)(1+\rho)\|v_\delta\|_2.
\end{align*}
Using the geometric series we reach to $\|i_{\text{har}}\|_2~\le~\frac{\sigma_{\text{max}}(M)}{V_0^2}\frac{(2+\rho)(1+\rho)}{1-\|\xi\|_\infty} \|v_\delta\|_2$, and using the minimum allowable voltage the following final expression can be obtained:
\begin{align*}
    \|i_{\text{har}}\|_2\le\frac{\sigma_{\text{max}}(M)}{v_{\text{min}}^2}\frac{(1+\rho)(2+\rho)}{1-(2\rho+\rho^2)} \|v_\delta\|_2.
\end{align*}}

\subsection{Appendix: Proof Theorem \ref{thm:GFT_nonlinear}}
\label{appendix:proofcor}
Recall condition \eqref{eqn:cor1}, $\forall s=\textup{j}\omega$ with $\omega\in[0,\infty)$
\begin{align*}
      -\left(\operatorname{SRG}(y_{l}(s))+\operatorname{SRG}(\widehat{y_{cp}})\right) \cap\tau\operatorname{SRG}(Y_{c}(s))= \emptyset. 
\end{align*}

As $\widehat{y_{cp}}$ satisfies the chord property by Lemma \ref{lemma:approx_SRG}, then
\begin{align}
    \operatorname{SRG}(y_{l}(s))+\operatorname{SRG}(\widehat{y_{cp}})\supset\operatorname{SRG}(y_{l}(s)+\widehat{y_{cp}}).    \label{eqn:proofunitary_step3}
\end{align}
 
Using Property \ref{property:sum} and Lemma \ref{lemma:approx_SRG}, we reach to
\begin{align}
    \operatorname{SRG}(y_{l}(s)+\widehat{y_{cp}})\supset\operatorname{SRG}(y_{l}(s)+{y_{cp}}).
    \label{eqn:proofunitary_stepn+1}
\end{align}

Given Property \ref{property:inversion} and using \eqref{eqn:Zgrid_nonlinear}, we can rewrite \eqref{eqn:proofunitary_stepn+1} as 
 \begin{align}
  \operatorname{SRG}(y_{l}(s)+{y_{cp}})=\operatorname{SRG}(Y_{\text{grid}}),
   \label{eqn:proofunitary_step4}
 \end{align} 
on the other hand, recall  $\Tilde{Y}_c(s)=  J(-\theta)Y_c(s)J(\theta)$. The SRG of the $\Tilde{Y}_c(s)$ is
\begin{align}
     \operatorname{SRG}(\Tilde{Y}_c(s))= \text{SRG}( J(-\theta)Y_c(s)J(\theta)). \label{eqn:proofunitary_step1}
\end{align}

The SRG operator is unitarily invariant \cite[Thm 1]{pates2021scaled}. In other words, given any linear operator $U \in \mathbb{C}^{n\times n}$ such that $UU^* = U^*U = I_n$, then $\operatorname{SRG}(U^*A U) = \operatorname{SRG}(A)$. Therefore, as $J(-\theta)J(\theta)=J(\theta)J(-\theta)=I_m$ for any arbitrary $\theta$, we can rewrite \eqref{eqn:proofunitary_step1} as
\begin{align}
  \operatorname{SRG}(\Tilde{Y}_c(s))=  \operatorname{SRG}({Y}_c(s)).\label{eqn:proofunitary_step2}
\end{align}

Consequently, we can rewrite \eqref{eqn:cor1} using \eqref{eqn:proofunitary_step4} and \eqref{eqn:proofunitary_step2} as
\begin{align*}
     \operatorname{SRG}(Y_{\text{grid}}) \cap -\tau \operatorname{SRG}(\Tilde{Y}_c(s)) = \emptyset,\;\; \forall \tau \in (0,1],
\end{align*}
$\forall s=\textup{j}\omega$ with $\omega\in[0,\infty)$, which is the stability condition in Theorem \ref{thm:GFTLTI}.



\subsection{Grid Following and Grid Forming Controller Derivation} \label{appendix:GFL_GFM} 

This appendix shows the detailed parameters of the GFL and GFM controllers used in Subsection \ref{subsec:scenario1} and \ref{subsec:Scenario2_nonlinear}\cite{Simplus}.  The measured system variables (shown in red in Figs. \ref{fig:GFL_controller} and \ref{fig:GFM_controller}) are $i_d$, $i_q$, $v_{dc}$, $v_d$, $v_q$, and $Q$, representing the $d$-axis current, $q$-axis current, DC voltage, $d$-axis voltage, $q$-axis voltage, and reactive power, respectively. Variables marked with $^*$ denote their reference values.   The control diagrams for the PLL-based GFL controller are shown in Fig.\ref{fig:GFL_controller}. 
Fig.~\ref{fig:GFL_controller}(a) shows the current controller where $PI_{cc}(s)=K_{p_{i}}+\frac{K_{i_{i}}}{s}$, where $K_{p_{i}}=X f_i/f_{base}$ and $K_{i_{i}}=K_{p_{i}} 2\pi f_i/4  $. The DC-link controller is depicted in Fig.~\ref{fig:GFL_controller}(b) where $PI_{dc}(s)=K_{p_{dc}}+\frac{K_{i_{dc}}}{s}$, and $K_{p_{dc}}=v_{dc} C_{dc} 2\pi f_{vdc}$, and $K_{i_{dc}}=K_{p_{dc}} 2\pi f_{dc}/4$. Finally, the PLL PI controller is depicted in Fig.~\ref{fig:GFL_controller}(c) where $PI_{pll}(s)=K_{p_{pll}}+\frac{K_{i_{pll}}}{s}$, where $K_{p_{pll}}= 2\pi f_{\text{pll}}$ and $K_{i_{pll}}=K_{p_{pll}} 2\pi f_{\text{pll}}/4$. where $f_{dc}$, $f_{\text{pll}}$, $f_i$ refer to the dc-link bandwidth, PLL cutoff frequency, and current bandwidth, respectively.
\begin{figure}[ht]
\begin{subfigure}[b]{0.15\textwidth}
    \centering
    \begin{tikzpicture}[scale=0.8, every node/.style={transform shape}]
\draw (9.4,4.35) rectangle (10.6,3.65);
\node at (10,4) {$PI_{cc}(s)$};
\draw[fill = white] (9,4) circle [radius=0.05];

\draw[-latex, line width = .5 pt] (10.6,4) -- (11.2,4) ;
\draw[-latex, line width = .5 pt] (9.05,4) -- (9.4,4);
\draw[-latex, line width = .5 pt] (8,3.5)--(9,3.5)-- (9,3.95);
\draw[-latex, line width = .5 pt] (8,4) -- (8.95,4);
\node at (8.812,3.8444) {\tiny$-$};
\node at (8.812,4.1664) {\tiny$+$};
\node at (10.85,4.2) {$e_d$};
\node at (8,4.2) {$i_d^*$};
\node at (8,3.7) {\tbr{$i_d$}};

\draw (9.4,5.35) rectangle (10.6,4.65);
\node at (10,5) {$PI_{cc}(s)$};
\draw[fill = white] (9,5) circle [radius=0.05];
\draw[-latex, line width = .5 pt] (10.6,5) -- (11.2,5) ;
\draw[-latex, line width = .5 pt] (9.05,5) -- (9.4,5);
\draw[-latex, line width = .5 pt] (8,4.5)--(9,4.5)-- (9,4.95);
\draw[-latex, line width = .5 pt] (8,5) -- (8.95,5);
\node at (8.812,4.8444) {\tiny$-$};
\node at (8.812,5.1664) {\tiny$+$};
\node at (10.85,5.2) {$e_q$};
\node at (8,5.2) {$i_q^*$};
\node at (8,4.7) {\tbr{$i_q$}};
\end{tikzpicture}
    \caption{}
\end{subfigure}
\begin{subfigure}[b]{0.15\textwidth}
    \centering
    \begin{tikzpicture}[scale=0.8, every node/.style={transform shape}]
\draw (9.4,4.35) rectangle (10.6,3.65);
\node at (10,4) {$PI_{dc}(s)$};
\draw[fill = white] (9,4) circle [radius=0.05];

\draw[-latex, line width = .5 pt] (10.6,4) -- (11.2,4) ;
\draw[-latex, line width = .5 pt] (9.05,4) -- (9.4,4);
\draw[-latex, line width = .5 pt] (8,3.5)--(9,3.5)-- (9,3.95);
\draw[-latex, line width = .5 pt] (8,4) -- (8.95,4);
\node at (8.812,3.8444) {\tiny$-$};
\node at (8.812,4.1664) {\tiny$+$};
\node at (10.85,4.2) {$i_d^*$};
\node at (8,4.2) {$v_{dc}^*$};
\node at (8,3.4) {\tbr{$v_{dc}$}};
\end{tikzpicture}
\caption{}
\end{subfigure}
\begin{subfigure}[b]{0.15\textwidth}
\begin{tikzpicture}[scale=0.8, every node/.style={transform shape}]
\draw (9.4,4.35) rectangle (10.6,3.65);
\node at (10,4) {$PI_{pll}(s)$};
\draw[fill = white] (11.2,4) circle [radius=0.05];
\draw (11.6,4.35) rectangle (12,3.65);
\node at (11.8,4) {$\frac{1}{s}$};

\draw[-latex, line width = .5 pt] (10.6,4) -- (11.15,4) ;
\draw[-latex, line width = .5 pt] (8.9,4) -- (9.4,4);
\draw[-latex, line width = .5 pt] (11.2,3.5)-- (11.2,3.95);
\draw[-latex, line width = .5 pt] (11.25,4) -- (11.6,4) ;
\draw[-latex, line width = .5 pt] (12,4) -- (12.4,4) ;

\node at (11,3.8444) {\tiny$+$};
\node at (11,4.1664) {\tiny$+$};
\node at (12.2,4.2) {$\theta$};
\node at (11,3.4) {$\omega_{0}$};
\node at (9,4.3) {\tbr{$v_{q}$}};
\end{tikzpicture}
    \caption{}
\end{subfigure}
    \caption{(a) Current, (b) dc-link  and (c) PLL controller.}
\label{fig:GFL_controller}
\end{figure}
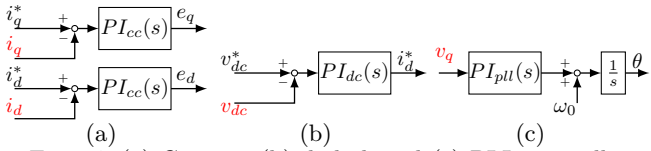

For the GFM controller, the control diagrams are shown in Fig. \ref{fig:GFM_controller}. Fig.~\ref{fig:GFM_controller}(a) shows the current controller where $PI_{cc}(s)=K_{p_{i}}+\frac{K_{i_{i}}}{s}$, where $K_{p_{i}}=X f_{idq}/f_{base}$ and $K_{i_{i}}=K_{p_i} 2\pi f_{idq}/4  $. Fig.~\ref{fig:GFM_controller}(b) shows the voltage controller where $PI_{vc}(s)=K_{p_{vc}}+\frac{K_{i_{vc}}}{s}$, where $K_{p_{vc}}= 2\pi f_{vdq}/4$ and $K_{i_{vc}}=K_{p_{vc}} 2\pi f_{vdq}/4$. Fig.~\ref{fig:GFM_controller}(c) shows the voltage regulation controller with parameter $P_{vr}$. Finally, Fig.~\ref{fig:GFM_controller}(d) shows the swing equation implemented in the GFM controller, using the well-known virtual synchronous generator strategy, where $J$ and $D$ represent the inertia and damping, respectively.

\begin{figure}[ht]
\begin{subfigure}[b]{0.1\textwidth}
    \centering
    \begin{tikzpicture}[scale=0.8, every node/.style={transform shape}]
\draw (9.4,4.35) rectangle (10.6,3.65);
\node at (10,4) {$PI_{cc}(s)$};
\draw[fill = white] (9,4) circle [radius=0.05];

\draw[-latex, line width = .5 pt] (10.6,4) -- (11.2,4) ;
\draw[-latex, line width = .5 pt] (9.05,4) -- (9.4,4);
\draw[-latex, line width = .5 pt] (8.4,3.5)--(9,3.5)-- (9,3.95);
\draw[-latex, line width = .5 pt] (8.4,4) -- (8.95,4);
\node at (8.812,3.8444) {\tiny$-$};
\node at (8.812,4.1664) {\tiny$+$};
\node at (10.85,4.2) {$e_d$};
\node at (8.4,4.2) {$i_d^*$};
\node at (8.4,3.7) {\tbr{$i_d$}};

\draw (9.4,5.35) rectangle (10.6,4.65);
\node at (10,5) {$PI_{cc}(s)$};
\draw[fill = white] (9,5) circle [radius=0.05];
\draw[-latex, line width = .5 pt] (10.6,5) -- (11.2,5) ;
\draw[-latex, line width = .5 pt] (9.05,5) -- (9.4,5);
\draw[-latex, line width = .5 pt] (8.4,4.5)--(9,4.5)-- (9,4.95);
\draw[-latex, line width = .5 pt] (8.4,5) -- (8.95,5);
\node at (8.812,4.8444) {\tiny$-$};
\node at (8.812,5.1664) {\tiny$+$};
\node at (10.85,5.2) {$e_q$};
\node at (8.4,5.2) {$i_q^*$};
\node at (8.4,4.7) {\tbr{$i_q$}};
\end{tikzpicture}
    \caption{}
\end{subfigure}
\begin{subfigure}[b]{0.1\textwidth}
    \centering
    \begin{tikzpicture}[scale=0.8, every node/.style={transform shape}]
\draw (9.4,4.35) rectangle (10.6,3.65);
\node at (10,4) {$PI_{vc}(s)$};
\draw[fill = white] (9,4) circle [radius=0.05];

\draw[-latex, line width = .5 pt] (10.6,4) -- (11.2,4) ;
\draw[-latex, line width = .5 pt] (9.05,4) -- (9.4,4);
\draw[-latex, line width = .5 pt] (8.4,3.5)--(9,3.5)-- (9,3.95);
\draw[-latex, line width = .5 pt] (8.4,4) -- (8.95,4);
\node at (8.812,3.8444) {\tiny$-$};
\node at (8.812,4.1664) {\tiny$+$};
\node at (10.85,4.2) {$i_d^*$};
\node at (8.4,4.2) {$v_d^*$};
\node at (8.4,3.7) {\tbr{$v_d$}};

\draw (9.4,5.35) rectangle (10.6,4.65);
\node at (10,5) {$PI_{vc}(s)$};
\draw[fill = white] (9,5) circle [radius=0.05];
\draw[-latex, line width = .5 pt] (10.6,5) -- (11.2,5) ;
\draw[-latex, line width = .5 pt] (9.05,5) -- (9.4,5);
\draw[-latex, line width = .5 pt] (8.4,4.5)--(9,4.5)-- (9,4.95);
\draw[-latex, line width = .5 pt] (8.4,5) -- (8.95,5);
\node at (8.812,4.8444) {\tiny$-$};
\node at (8.812,5.1664) {\tiny$+$};
\node at (10.85,5.2) {$i_q^*$};
\node at (8.4,5.2) {$v_q^*$};
\node at (8.4,4.7) {\tbr{$v_q$}};
\end{tikzpicture}
    \caption{}
\end{subfigure}
\begin{subfigure}[b]{0.16\textwidth}
    \centering
    \begin{tikzpicture}[scale=0.8, every node/.style={transform shape}]
\draw (9.4,4.35) rectangle (9.8,3.65);
\node at (9.6,4) {$\frac{1}{J}$};
\draw (10.2,4.35) rectangle (10.6,3.65);
\node at (10.4,4) {$\frac{1}{s}$};
\draw (11,4.35) rectangle (11.4,3.65);
\node at (11.2,4) {$\frac{1}{s}$};
\draw (9.7,3.1) rectangle (10.3,3.5);
\node at (10,3.3) {$D$};
\draw[fill = white] (9,4) circle [radius=0.05];
\draw[fill = white] (10.75,3.3) circle [radius=0.05];

\draw[-latex, line width = .5 pt] (9.8,4) -- (10.2,4) ;
\draw[-latex, line width = .5 pt] (10.6,4) -- (11,4) ;
\draw[-latex, line width = .5 pt] (11.4,4) -- (11.8,4) ;
\draw[-latex, line width = .5 pt] (9.05,4) -- (9.4,4);
\draw[-latex, line width = .5 pt] (9.7,3.3)--(9,3.3)-- (9,3.95); 
\draw[-latex, line width = .5 pt] (8.4,4.5)--(9,4.5)-- (9,4.05);
\draw[-latex, line width = .5 pt] (8.4,4) -- (8.95,4);
\draw[-latex, line width = .5 pt] (11.5,3.3) -- (10.8,3.3);
\draw[-latex, line width = .5 pt] (10.7,3.3) -- (10.3,3.3);
\draw[-latex, line width = .5 pt] (10.75,4)--(10.75,3.35);

\node at (8.812,3.8444) {\tiny$-$};
\node at (9.122,4.1664) {\tiny$-$};
\node at (8.812,4.1664) {\tiny$+$};
\node at (10.9,3.5) {\tiny$+$};
\node at (10.9,3.1664) {\tiny$-$};
\node at (11.6,4.2) {$\theta$};
\node at (11.3,3.1) {$\omega_0$};
\node at (8.4,4.2) {$Q^*$};
\node at (8.4,4.7) {\tbr{$Q$}};
\end{tikzpicture}
\caption{} 
\end{subfigure}
\begin{subfigure}[b]{0.12\textwidth}
    \centering
     \begin{tikzpicture}[scale=0.8, every node/.style={transform shape}]
\draw (9.4,4.35) rectangle (10,3.65);
\node at (9.7,4) {$P_{vr}$};
\draw[fill = white] (9,4) circle [radius=0.05];
\draw[fill = white] (10.55,4) circle [radius=0.05];

\draw[-latex, line width = .5 pt] (10,4) -- (10.5,4) ;
\draw[-latex, line width = .5 pt] (10.6,4) -- (11.2,4) ;
\draw[-latex, line width = .5 pt] (9.05,4) -- (9.4,4);
\draw[-latex, line width = .5 pt] (8,3.5)--(9,3.5)-- (9,3.95);
\draw[-latex, line width = .5 pt] (8,4) -- (8.95,4);
\draw[-latex, line width = .5 pt] (10.55,3.5) -- (10.55,3.95);
\node at (8.812,3.8444) {\tiny$-$};
\node at (8.812,4.1664) {\tiny$+$};
\node at (10.3912,3.8444) {\tiny$+$};
\node at (10.392,4.1664) {\tiny$+$};
\node at (10.85,4.2) {$v_d^*$};
\node at (10.75,3.6) {$e_d^*$};
\node at (8,4.2) {$Q^*$};
\node at (8,3.7) {\tbr{$Q$}};
\end{tikzpicture}
    \caption{} 
\end{subfigure}
    \caption{(a) Current, (b) Voltage, (c) Swing, and (d) Voltage-regulator controller.}
\label{fig:GFM_controller}
\end{figure}
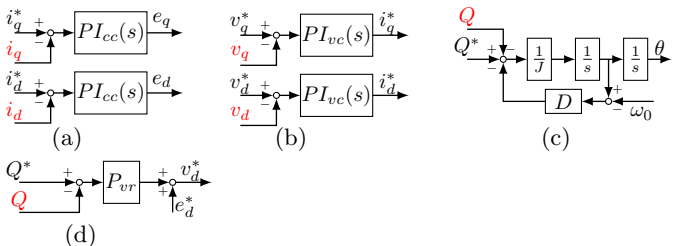

The parameters used for the \textit{Case study 1} and \textit{Case study 2} are shown in Table \ref{table:GFL_parameters}.

\begin{table}[h]
\begin{tabular}{l|ll|llll}
\hline \multicolumn{1}{c}{ Parameters}   &  \multicolumn{2}{c}{ Case study 1 }   &\multicolumn{4}{c}{ Case study 2 } \\ 
& GFL$_1$ &GFL$_2$& GFL& GFM &CPL$_1$&CPL$_2$\\
\hline$v_{dc}$ (p.u.)& 2.5 & 2.5 & 2.5 & -& -& -\\
$C_{dc}$ (p.u.)& 1.25 & 1.25 & 1.25 & - & -& -\\
$f_{dc}$ (Hz) & 10& 10& 10 & - & - & -\\
$f_{\text{pll}}$ (Hz) & 30& 70 & 30 & - & -  & -\\
$f_{i}$ (Hz) & 250 & 250 & 250 & - & - & - \\
$X$  (p.u.)& 0.01 & 0.01 & 0.01 & 0.001 & -& -\\
$R$  (p.u.)& 0 & 0 & 0 & 0& -& -\\
$D$ & - & - &  - &  0.05 & - & -\\
$J$ (Hz) & - & - & -  &  10$\pi$ &  -& -\\
$P_{vr}$ & - & - &  - &  0.05 &   -& -\\
$f_{vdq}$ (Hz) & - & - &  -  &  300  &  - & - \\
$f_{idq}$ (Hz) & - & - &  - &  600 &  -& -\\
$p_c$ (p.u.) & - & - &  -  &  -  &  0.1& 0.56\\
$q_c$ (p.u.) & - & - &  - &  - & 0.1& 0.1\\
\hline
\end{tabular}
    \caption{Converter parameters for Case studies in Section \ref{sec:conditions_linear} and Section \ref{sec:conditions_nonlinear}.}
\label{table:GFL_parameters}
\end{table}
\subsection{System Analysis Parameters}\label{appendix:system_analysis}

\subsubsection{IEEE 14-node system}
All GFL converters used in this case study follow the parameters listed in Table~\ref{table:GFL_system_parameters}. For Case~1, the converters at nodes 2, 3, 6, and 8 correspond to the column labeled GFL$_{2,3,6,8}$; their only difference lies in the PLL bandwidths, set to 45, 40, 35, and 30~Hz, respectively. For the unstable scenario, the parameters of the modified GFL converter at node~8 are shown under GFL$_{8}^{\mathrm{uns}}$. Line parameters and the synchronous machine data are taken from~\cite{Simplus}.

\subsubsection{IEEE 57-node system} 
The GFL converters used in this system also follow the parameter sets in Table~\ref{table:GFL_system_parameters}. For Case~1, the converters at nodes 2, 6, and 12 correspond to the column GFL$_{2,6,12}$. The GFM converter at node~8 is denoted GFM$_8$. For the unstable case, the modified GFL converter at node~12 is given by GFL$_{12}^{\mathrm{uns}}$. As in the previous system, line parameters and synchronous machine models are taken from~\cite{Simplus}.

\begin{table}[h]
\begin{tabular}{l|ll|llll}
\hline \multicolumn{1}{c}{ Parameters}   &  \multicolumn{2}{c}{ IEEE 14 node }   &\multicolumn{3}{c}{ IEEE 57 node } \\ 
& GFL$_{2,3,6,8}$ &GFL$_8^{\mathrm{uns}}$& GFL$_{2,6,12}$& GFM$_8$ &GFL$_{12}^{\mathrm{uns}}$\\
\hline$v_{dc}$ (p.u.)& 2.5 & 2.5 & 2.5 & -& 1\\
$C_{dc}$ (p.u.)& 1.25 & 1.25 & 1.25 & - & 1\\
$f_{dc}$ (Hz) & 10& 10& 10 & - & 5 \\
$f_{\text{pll}}$ (Hz) & 45, 40, 35, 30& 30 & 5, 10, 15 & - & 15  \\
$f_{i}$ (Hz) & 600 & 100 & 250 & - & 300  \\
$X$  (p.u.)& 0.01 & 0.01 & 0.01 & 0.001 & 0.01\\
$R$  (p.u.)& 0 & 0 & 0.01 & 0& 0.01\\
$D$ & - & - &  - &  0.002 & - \\
$J$ (Hz) & - & - & -  &  10$\pi$ &  -\\
$P_{vr}$ & - & - &  - &  5 &   -\\
$f_{vdq}$ (Hz) & - & - &  -  &  150  &  - \\
$f_{idq}$ (Hz) & - & - &  - &  300 &  -\\ 
\hline
\end{tabular}
    \caption{Converter parameters for Case studies in Section \ref{sec:system_analysis}.}
\label{table:GFL_system_parameters}
\end{table}
\subsection{Comparison Small gain-phase and passivity parameters}\label{appendix:comparison}
This appendix shows the controller parameters of the GFL used in Section \ref{sec:comparison}\cite{Wang2024Limitations,chen2022impedance}. The parameters are as follows: $ L $ is 0.1 p.u., $ \alpha_c $ is 10, $ \alpha_p $ is 0.35, $ R_a $ is 0.2 p.u., $ G_{ad} $ and $ G_{aq} $ are 2.0, $ K_v $ is 0, $ \alpha_f $ is 10, $ L_g $ is 0.6 p.u., $ R_g $ is 0.2 p.u., $ i_{d0} $ is 1 p.u., $ i_{q0} $ is 0 p.u., and $ E_0 $ is 1 p.u..
\end{document}